%% file: main.tex
\documentclass[conference]{IEEEtran}
\IEEEoverridecommandlockouts
 
\usepackage{cite}
\usepackage{amsmath,amssymb,amsfonts,amsthm}
\usepackage{algorithmic}
\usepackage{graphicx}
\usepackage{textcomp}
\usepackage{xcolor}
\usepackage{caption}
\usepackage{subcaption}
\usepackage{booktabs}
\usepackage{multirow}
\usepackage{pifont}
\usepackage{url}
\usepackage[colorlinks=true,allcolors=blue]{hyperref}

\IfFileExists{balance.sty}{\usepackage{balance}}{\newcommand{\balance}{}}
 
\graphicspath{{figures/}}

\IfFileExists{numbers.tex}{\input{numbers}}{%
  \typeout{WARNING: numbers.tex missing; run sim/plots/make_numbers.py}}

\newtheorem{theorem}{Theorem}
\newtheorem{lemma}[theorem]{Lemma}
\newtheorem{proposition}[theorem]{Proposition}
\newtheorem{corollary}[theorem]{Corollary}
\theoremstyle{definition}
\newtheorem{definition}[theorem]{Definition}

\theoremstyle{remark}
\newtheorem{remark}[theorem]{Remark}

\newcommand{\Rev}{\mathcal{R}}        
\newcommand{\Pap}{\mathcal{P}}        
\newcommand{\Aff}{\mathbf{A}}         
\newcommand{\Bid}{\mathbf{B}}        
\newcommand{\Asg}{\mathbf{X}}      
\newcommand{\actI}{\textsf{I}}      
\newcommand{\actH}{\textsf{H}}        

\newcommand{\Prob}{\mathbb{P}}
\newcommand{\dstar}{d^{\star}}
\newcommand{\qstar}{q^{\star}}
\newcommand{\deltastar}{\delta^{\star}}
\newcommand{\mstar}{m^{\star}}

\begin{document}

\title{Looking for Bidding Teammates: A Game-Theoretic\\
Model of Stranger Collusion in Peer Review}

\author{ 
\IEEEauthorblockN{Jinming Xing}
\IEEEauthorblockA{
\textit{North Carolina State University}\\
jxing6@ncsu.edu}
\and
\IEEEauthorblockN{Charlotte Brian}
\IEEEauthorblockA{
\textit{Independent Researcher}\\
immjomrtal@gmail.com}
}

% \author{\IEEEauthorblockN{Anonymous Submission}}

\maketitle

\input{sections/00_abstract}


\input{sections/01_introduction} 
\input{sections/02_background} 
\input{sections/03_model}
\input{sections/04_equilibrium}  
\input{sections/05_taxonomy}
\input{sections/06_experiments}
\input{sections/07_discussion}
\input{sections/08_conclusion}

\balance
\bibliographystyle{IEEEtran}
\bibliography{refs}

\appendices
\input{sections/09_appendix}

\end{document}

%% file: numbers.tex
% Generated by sim/plots/make_numbers.py. Do not edit by hand.
\newcommand{\dStar}{0.113}

\newcommand{\dTransition}{0.079}
\newcommand{\qStar}{0.571}
\newcommand{\qTransition}{0.627}

\newcommand{\nAgents}{20{,}000}
\newcommand{\collusionHand}{68.2}
\newcommand{\ringSizeHand}{4.6}
\newcommand{\effortCostHand}{0.40}
\newcommand{\collusionMachine}{91.9}

\newcommand{\effortCostMachine}{0.08}
\newcommand{\collusionAuthor}{95.2}
\newcommand{\ringSizeAuthor}{8.6}

\newcommand{\succNaive}{0.312}
\newcommand{\fOneNaiveFreq}{0.000}

\newcommand{\fOneNaiveCycle}{0.156}

\newcommand{\fOneCamoCycle}{0.322}

\newcommand{\fOneDistCycle}{0.032}

\newcommand{\succAff}{0.859}

\newcommand{\succComb}{0.673}

\newcommand{\fOneCombWorst}{0.216}
\newcommand{\cellNaiveFreq}{\textbf{0.000}}
\newcommand{\cellCamoFreq}{\textbf{0.000}}
\newcommand{\cellDistFreq}{0.079}
\newcommand{\cellAffFreq}{\textbf{0.000}}
\newcommand{\cellCombFreq}{0.127}
\newcommand{\cellNaiveDense}{\textbf{0.069}}
\newcommand{\cellCamoDense}{0.164}
\newcommand{\cellDistDense}{0.110}
\newcommand{\cellAffDense}{\textbf{0.069}}
\newcommand{\cellCombDense}{0.216}
\newcommand{\cellNaiveCycle}{0.156}
\newcommand{\cellCamoCycle}{0.322}
\newcommand{\cellDistCycle}{\textbf{0.032}}
\newcommand{\cellAffCycle}{0.156}
\newcommand{\cellCombCycle}{0.115}
\newcommand{\cellNaiveText}{0.047}
\newcommand{\cellCamoText}{\textbf{0.030}}
\newcommand{\cellDistText}{0.049}
\newcommand{\cellAffText}{0.085}
\newcommand{\cellCombText}{0.063}
\newcommand{\rhoZero}{0.003}
\newcommand{\evasionFPR}{5}
\newcommand{\evasionSeeds}{8}
\newcommand{\colluderFrac}{5}
\newcommand{\fOneBestOverall}{0.322}
\newcommand{\boostDyad}{3.2}
\newcommand{\acceptDyad}{28.2}
\newcommand{\acceptBase}{25.0}
\newcommand{\bestRingSize}{4}
\newcommand{\boostBestRing}{4.8}

\newcommand{\dispCamoFive}{51}

\newcommand{\dispDistFive}{13}

\newcommand{\dispCombFive}{70}
\newcommand{\dispCombTen}{129}

\newcommand{\qualityDrop}{0.002}
\newcommand{\qualityPTenDrop}{0.007}
\newcommand{\impactSeeds}{8}
\newcommand{\inflation}{2.0}

\newcommand{\residNoneNaive}{0.312}
\newcommand{\residNoneCamo}{0.271}
\newcommand{\residNoneDist}{0.255}
\newcommand{\residNoneComb}{0.673}
\newcommand{\qualityCycle}{100.0}

\newcommand{\residCycleNaive}{0.257}
\newcommand{\residCycleCamo}{0.226}
\newcommand{\residCycleDist}{0.249}
\newcommand{\residCycleComb}{0.665}
\newcommand{\qualityRand}{96.8}

\newcommand{\residRandComb}{0.311}
\newcommand{\qualityBoth}{96.8}

\newcommand{\residBothComb}{0.308}
\newcommand{\rhoNone}{0.673}
\newcommand{\kMinNone}{2.06}
\newcommand{\kMaxNone}{5.96}

\newcommand{\rhoRand}{0.311}
\newcommand{\kMinRand}{6.19}
\newcommand{\kMaxRand}{5.42}

\newcommand{\sensDTransHomog}{0.060}

\newcommand{\sensFracHomog}{100.0}
\newcommand{\sensDTransSpreadLo}{0.076}
\newcommand{\sensDTransSpreadHi}{0.081}
\newcommand{\sensQTransSpreadLo}{0.573}
\newcommand{\sensQTransSpreadHi}{0.650}
\newcommand{\sensBSpreadHi}{0.70}

\newcommand{\sensRateLo}{15}
\newcommand{\sensRateHi}{40}
\newcommand{\sensBoostRateLo}{2.4}
\newcommand{\sensBoostRateHi}{4.7}
\newcommand{\sensDispRateLo}{39}
\newcommand{\sensDispRateHi}{68}
\newcommand{\sensAcceptRateLo}{17.4}
\newcommand{\sensAcceptRateHi}{43.4}
\newcommand{\sensLambdaLo}{1}
\newcommand{\sensLambdaHi}{20}
\newcommand{\nStarLambdaLo}{4.9}
\newcommand{\nStarLambdaHi}{0.2}
\newcommand{\sensCeTransSpreadLo}{0.280}
\newcommand{\sensCeTransSpreadHi}{0.472}
\newcommand{\sensRegimeHandLo}{51.9}
\newcommand{\sensRegimeHandHi}{87.4}

\newcommand{\gridFPRLo}{1}
\newcommand{\gridFPRHi}{20}

\newcommand{\gridBestAnywhere}{0.567}
\newcommand{\gridBestCombAnywhere}{0.228}
\newcommand{\gridSeeds}{8}
\newcommand{\bidsHonest}{4.8}
\newcommand{\camoShareLo}{0.078}
\newcommand{\camoShareHi}{0.258}
\newcommand{\camoBidsHi}{13.4}
\newcommand{\camoRatioHi}{8}
\newcommand{\camoSuccLo}{0.232}
\newcommand{\camoSuccHi}{0.296}
\newcommand{\camoDenseLo}{0.096}
\newcommand{\camoDenseHi}{0.245}
\newcommand{\camoCycleLo}{0.205}
\newcommand{\camoCycleHi}{0.416}
\newcommand{\camoFreqHi}{0.000}
\newcommand{\gridFreqDistPrec}{1.000}
\newcommand{\gridFreqDistRec}{0.041}
\newcommand{\gridFreqDistFOne}{0.079}

\newcommand{\gridFreqCombRec}{0.068}
\newcommand{\gridFreqCombFOne}{0.127}
\newcommand{\gridCycleCamoPrec}{0.530}
\newcommand{\gridCycleCamoRec}{0.232}

\newcommand{\cycleFOneRingSixDeep}{0.095}
\newcommand{\cycleHonestShallow}{0.068}
\newcommand{\cycleHonestDeep}{0.180}
\newcommand{\cycleHonestShallowPct}{6.8}

\newcommand{\cycleSecShallow}{0.01}
\newcommand{\cycleSecDeep}{0.61}
\newcommand{\cycleBoundHi}{4}
\newcommand{\cycleRingHi}{6}
\newcommand{\cycleRingThreeShallow}{0.032}
\newcommand{\cycleRingThreeMatched}{0.248}
\newcommand{\cycleRingThreeDeep}{0.144}
\newcommand{\cycleRingFourMatched}{0.141}

%% file: sections/00_abstract.tex
\begin{abstract}
Paper bidding is the entry point to reviewer assignment at large computer
science conferences: reviewers declare interest in submissions, and an optimizer
combines them with automated affinity scores. Reviewers who have
never met recruit each other on public social platforms, exchange submission
identifiers, bid on each other's papers, and reciprocate with inflated scores.
Existing collusion models take the group as given and assume its members are established colleagues who already trust one another; open
recruitment removes that assumption together with the mechanism that made the
arrangement work. We give the first game-theoretic model of collusion
\emph{formation} in peer review, the \emph{Mutual Bidding Dilemma}: a
four-stage game covering recruitment, the exchange of identifiers under the risk
of being reported, bidding that neither party can verify, and reciprocal
reviewing. The model predicts the arrangement cannot form: once a colluder is assigned to
a partner's paper, writing the inflated review is pure cost, because the benefit
they care about is delivered by the partner's separate decision about their own
paper. Reciprocation is never individually rational, for any payoffs, and the
arrangement unwinds to the recruitment post. What
closes the gap is enforcement rather than incentives: authors always see the reviews of their own paper, deadlines recur every few months, and the recruitment group remembers who
reciprocated. We give the condition under which reciprocal inflation becomes
sustainable, the detection rate above which no partnership survives, and show
review-writing effort enters both. On a calibrated
end-to-end conference simulation, no detector we test exceeds
$F_1 = \fOneBestOverall$ against an attacker who camouflages bids, distributes
them around a ring, and manipulates affinity; a two-person arrangement is worth
$\boostDyad$ percentage points of acceptance probability; and the harm is
distributional rather than aggregate, with $\dispCombFive$ honest papers
displaced while mean accepted quality moves by $\qualityDrop$, so no summary
statistic a venue publishes reveals it. Randomized assignment is
the one defense that reaches the enforcement itself, making a partner who never
bid indistinguishable from one who bid and lost.
\end{abstract}

%% file: sections/01_introduction.tex
\section{Introduction}
\label{sec:intro}

Machine learning conferences have outgrown the reviewer pool that serves them.
NeurIPS received roughly 2{,}400 submissions a decade ago and more than
21{,}000 in the most recent cycle, with ICLR, ICML and AAAI on similar
trajectories~\cite{ref:conf-growth}, and the strain on review quality is well
documented~\cite{ref:shah-survey,ref:nips-experiment}. No
program chair can assign reviewers by hand at that scale, so assignment is
delegated to an optimizer maximizing total reviewer-paper affinity subject to
load and coverage constraints~\cite{ref:tpms,ref:assignment-opt}. It takes two
inputs: affinity scores computed from text~\cite{ref:tpms,ref:specter}, and bids
in which reviewers declare which submissions they want. Bids exist because
affinity scores are imperfect and because reviewers who receive papers they
asked for write better reviews. They are also unpriced and unverified: a
reviewer pays nothing to declare interest in a paper they have none in, and no
part of the pipeline tests a declared interest against a genuine one.

\subsection{The phenomenon}

That gap is now being exploited openly. On public social platforms, including
WhatsApp groups, Telegram channels, and X, reviewers post messages whose subject
lines say, in effect, \emph{looking for bidding teammates} for a named
conference, with the terms stated in the body: I bid on your submission, you bid
on mine, and whoever is assigned returns a high score.
Fig.~\ref{fig:screenshots} shows three such threads with a selection of the
replies each drew. The posts name the venue, which we redact, and usually the
subject area, and they invite strangers to reply.

\begin{figure}[htbp]
\centering
\includegraphics[width=\columnwidth]{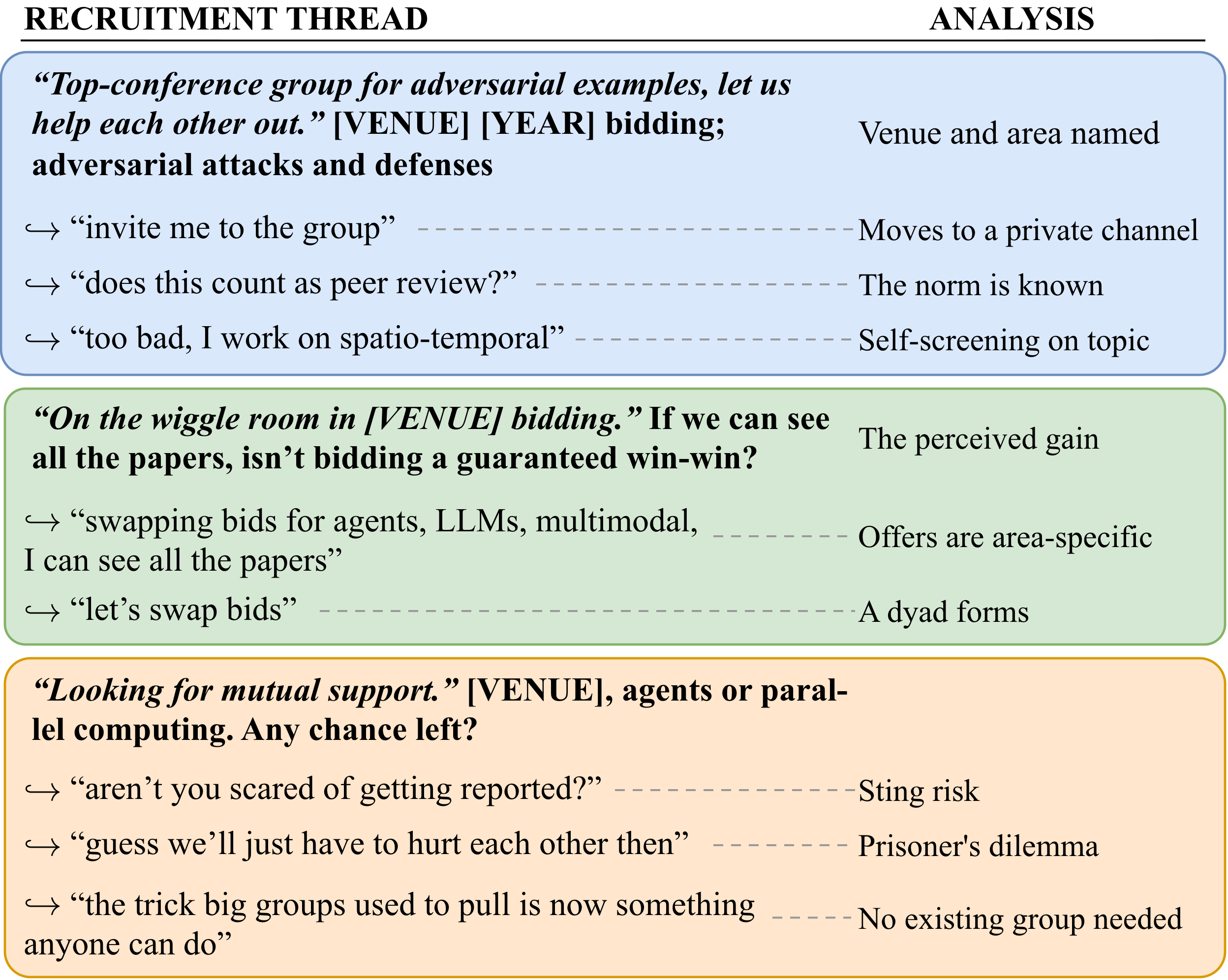}
\caption{Three public recruitment threads soliciting reciprocal bidding
partners, with selected replies ($\hookrightarrow$). We typeset the text rather
than reproduce screenshots, because the platform is identifiable from its
interface; posts are translated and lightly paraphrased so that no thread can
be recovered by searching a quoted string. No handles, avatars, timestamps,
group identifiers, or venue names appear. The right column is our annotation,
naming what each line instantiates in the model of \S\ref{sec:model}.}
\label{fig:screenshots}
\end{figure}

This is observed behavior rather than a hypothetical threat model: the posts are
public, easy to find, and recur every submission cycle. Organized manipulation
of computer science peer review is not itself new, since collusion rings among
acquainted researchers were described publicly several years
ago~\cite{ref:collusion-ring-report}. What is new is that recruitment has moved
into the open.

The replies are informative beyond confirming that strangers answer. Responders
sort themselves by subject area, which is what makes a partner's paper
reviewable at all and which we impose on the simulator in
\S\ref{sec:experiments}; they raise the risk of being reported, which is the
sting risk of Stage~2; and one of them states the reviewing-stage dilemma
outright, before we model it.

\subsection{Why open recruitment is a different problem}

The peer review collusion literature is small and almost entirely the product of
one line of work~\cite{ref:jecmen2020,ref:jecmen2022,ref:jecmen2023,
ref:jecmen2024,ref:hsieh2025}. Across it the colluding group is taken as given:
reviewers who already know each other, typically lab mates or frequent
co-authors, whose problem is to be assigned to each other's papers without being
caught. Detection and mitigation are then posed as questions about the bidding
matrix such a group produces. Parts of that literature are game-theoretic, but
the game they analyze begins after the group exists. None asks how the group
comes into being, and until recently there was little reason to: colluders were
colleagues, and the relationship predated the conference. Open recruitment makes
formation the interesting part, and formation is what we model.

Recruitment among strangers breaks the premise that makes that framing work. Two
people who met an hour ago in a group chat have no shared history, no
institutional bond, and no expectation of future contact. They are also, in the
formal sense that conflict-of-interest policies use, not in conflict with one
another, so no automated check separates them. Four problems follow that
colleagues do not face. \emph{First-mover
exposure}: somebody must send a submission identifier first, revealing
authorship, intent, and an accusation ready to forward to a program chair.
\emph{Sting risk}: the reply may come from a genuine partner, a curious
bystander, or someone collecting evidence. \emph{Unverifiable bidding}: bids are
private to the conference system, so failure to be assigned is consistent both
with an honest bid that lost and with no bid at all. \emph{Absent repeated-game
trust}: colleagues are disciplined by an ongoing relationship and strangers are
not.

These frictions pose a question distinct from the one the literature has
addressed. That reviewers want higher scores on their own submissions accounts
for the demand for such arrangements; it does not account for their stability.
Consider a colluder already assigned to a partner's paper. The benefit at stake,
a favorable review on their own submission, is produced by the partner's
separate decision about a separate paper, and nothing the colluder writes
affects it. Writing the inflated review is pure cost at the moment of writing:
the effort of making it survive the discussion phase, plus whatever sanction
risk it carries. Reporting honestly dominates, both parties anticipate as much,
and the anticipation propagates backwards through bidding, identifier exchange,
and the recruitment post itself. \S\ref{sec:equilibrium} shows this holds for
every value of the payoffs, including arbitrarily large ones: raising the value
of a publication makes collusion more attractive to enter without making
reciprocation rational once entered, so the operative constraint is enforcement
rather than incentives.

\subsection{Why current defenses do not settle the question}

Existing defenses act on the mechanics of assignment and therefore leave that
constraint untouched. Cycle-free assignment forbids two-cycles in the assignment
graph~\cite{ref:cyclefree}, which a three-party ring evades by construction.
Randomized assignment caps the probability that any target pair is
matched~\cite{ref:jecmen2020}, lowering the expected return on collusion without
changing its sign, at a documented cost in assignment quality. Detection applied
to bidding data loses most of its power once colluders mix collusive bids with
genuine ones~\cite{ref:jecmen2024}, and text-matching defenses address a
separate attack surface in which affinity scores rather than bids are
manipulated~\cite{ref:hsieh2025}. Each constrains what the optimizer may output;
none reaches the relationship that produced its inputs. A defense lowering the
odds of a collusive assignment is priced, in the literature proposing it, purely
as a reduction in the attacker's expected return, and
\S\ref{sec:equilibrium} and \S\ref{sec:e4} show that it also destroys the
colluders' ability to monitor each other.

\subsection{Contributions}

\begin{enumerate}
  \item \textbf{The Mutual Bidding Dilemma.} We give the first game-theoretic
  model of collusion \emph{formation} in peer review, and the first to treat the
  colluders as strangers rather than as colleagues who already trust one another
  (Table~\ref{tab:related}). The model is a four-stage extensive-form game
  covering recruitment signaling, information exchange under sting risk, bidding
  as a hidden action, and reciprocal reviewing (\S\ref{sec:model}). It is
  addressed to a phenomenon the literature has not modeled at all: open
  recruitment of bidding partners on public social platforms.

  \item \textbf{An unraveling result and its resolution.} The reviewing stage is
  a Prisoner's Dilemma for every parameter setting, so the one-shot game has a
  unique subgame-perfect equilibrium without collusion
  (Theorem~\ref{thm:unravel}), and observed collusion is evidence about
  enforcement technology rather than payoff magnitudes. We identify that
  technology and show that reciprocal inflation is sustainable exactly above a
  discount threshold $\deltastar$ (Theorem~\ref{thm:sustain}), that the first
  mover reveals an identifier only above a trust threshold $\qstar$
  (Proposition~\ref{prop:trust}), and that no discount factor sustains collusion
  once per-review detection exceeds $\dstar$ (Theorem~\ref{thm:dstar}). Effort
  cost enters $\deltastar$ directly, turning the claim that language models
  amplify collusion into a comparative static (\S\ref{sec:equilibrium}).

  \item \textbf{Attack taxonomy and evasion analysis.} Four bidding and four
  reviewing strategies mapped against four detector families, with simulation
  showing that camouflaged and distributed bidding together defeat every
  single-signal detector we test (\S\ref{sec:taxonomy}).

  \item \textbf{Impact quantification and defense implications.} A calibrated
  end-to-end conference simulation measuring the acceptance boost colluders
  obtain, the honest papers they displace, and the residual harm under existing
  defenses, translating the equilibrium conditions into levers a program chair
  controls (\S\ref{sec:experiments}, \S\ref{sec:discussion}).
\end{enumerate}

%% file: sections/02_background.tex
\section{Background and Related Work}
\label{sec:background}

\subsection{The conference review pipeline}

A large conference reviews in six stages: submission, bidding, assignment,
reviewing, discussion and meta-review, and decision. Two are exploitable by the
manipulation we study, and a third makes it enforceable.

\textbf{Bidding.} Reviewers browse titles and abstracts, usually filtered by
subject area, and record a preference over an effective choice set of
\emph{eager}, \emph{willing}, \emph{neutral}, \emph{not willing}, and
\emph{conflict}. A bid is cheap talk in the technical sense: submitting or
misreporting it costs nothing, and the system cannot test it against the
reviewer's actual interest. Participation is uneven, and many submissions
receive very few positive bids. That sparsity is what
gives a collusive bid its leverage.

\textbf{Assignment.} The assignment stage combines bids with an affinity score
$A_{ij}$ estimating how well reviewer $i$ matches paper $j$, computed from
TF-IDF similarity between the reviewer's publication record and the
submission~\cite{ref:tpms} or, in current deployments, from document embeddings
such as SPECTER~\cite{ref:specter,ref:specter2} wired into the venue's matching
service. The optimizer solves
\begin{equation}
\label{eq:assignment}
\begin{aligned}
\max_{\Asg \in \{0,1\}^{M\times N}} \;\; & \sum_{i,j} w(A_{ij}, B_{ij})\, X_{ij} \\[-2pt]
\text{s.t.} \;\; & \textstyle\sum_j X_{ij} \le L_i \;\; \forall i, \quad
                   \textstyle\sum_i X_{ij} = k \;\; \forall j, \\[-2pt]
                 & X_{ij} = 0 \;\; \forall (i,j) \in \mathcal{C},
\end{aligned}
\end{equation}
following~\cite{ref:taylor-assignment,ref:peerreview4all}, where $L_i$ is
reviewer $i$'s load cap, $k$ the required reviews per paper, and $\mathcal{C}$
the conflict-of-interest set. The
weight $w(\cdot)$ is monotone in both arguments and in most deployments gives
bids substantial influence, so an eager bid can outweigh a sizable affinity
deficit. That influence is what makes bidding worth manipulating.

\textbf{Discussion.} Most venues release initial reviews to the co-reviewers of
the same paper and open a period in which scores may be revised, and authors see
the reviews of their own paper. This stage is normally analyzed as a quality
mechanism, alongside score calibration and reviewer
disagreement~\cite{ref:stelmakh-bias}.
\S\ref{sec:equilibrium} shows it also functions as a monitoring device,
necessary to any explanation of why stranger collusion is stable.

\subsection{Manipulation, mitigation, and detection}

Jecmen et al.~\cite{ref:jecmen2020} introduced randomized reviewer assignment as
a mitigation: capping the probability that any reviewer-paper pair is matched
bounds the success probability of a collusive attempt, at a measured cost in
assignment quality of 5--15\%. A follow-up~\cite{ref:jecmen2022}
shows expertise matching and manipulation robustness cannot be maximized
simultaneously. Cycle-free reviewing~\cite{ref:cyclefree} instead forbids cycles
of length up to $k$ in the assignment graph; at $k=2$ this rules out the direct
reciprocity the recruitment posts describe, but the constraint is expensive for
larger $k$ and a ring of size $k+1$ evades it by construction. Randomized
transparency, releasing a random subset of bids after the conference closes, is a
further deterrent; in our framework it acts
on the detection probability, which \S\ref{sec:equilibrium} shows is the correct
target. All of these take the colluding set as given.

The detection side is less encouraging. Jecmen et al.~\cite{ref:jecmen2023}
released a simulated dataset of malicious bidding, and a subsequent
study~\cite{ref:jecmen2024} found that most off-the-shelf fraud detection
algorithms fail on it once colluders camouflage, placing collusive bids inside a
larger set of genuine bids on topically similar papers. Our experiments
reproduce that finding and extend it to strategy combinations the original
evaluation did not consider. Separately, Hsieh et al.~\cite{ref:hsieh2025}
showed the text-matching component is manipulable independently of bidding:
curating a reviewer's public publication list, or inserting background sentences
matched to a target reviewer's corpus, moves the affinity score enough to make
assignment likely. The two manipulations are complements, since raising
$A_{ij}$ raises the probability that a collusive bid succeeds, which enters our
payoffs through $\rho$ and by Theorem~\ref{thm:sustain} makes collusion easier
to sustain.

\begin{table}[tbp]
\centering
\caption{Positioning relative to prior work. \emph{Formation} means the paper
models how a collusive group comes into being rather than taking it as given;
\emph{strangers} means it does not assume prior trust between colluders;
\emph{impact} means it quantifies downstream effects on acceptance decisions.}
\label{tab:related}
\footnotesize
\setlength{\tabcolsep}{3pt}
\begin{tabular}{@{}l@{\hspace{4pt}}ccccc@{}}
\toprule
& \rotatebox{60}{Formation} & \rotatebox{60}{Strangers}
& \rotatebox{60}{Game-theor.} & \rotatebox{60}{Impact}
& \rotatebox{60}{Multi-stage} \\
\midrule
Randomized assign.\ \cite{ref:jecmen2020}   & \ding{55} & \ding{55} & \checkmark & \ding{55} & \ding{55} \\
Manipulation tradeoffs \cite{ref:jecmen2022}  & \ding{55} & \ding{55} & \checkmark & \ding{55} & \ding{55} \\
Malicious bid data \cite{ref:jecmen2023}  & \ding{55} & \ding{55} & \ding{55} & \ding{55} & \ding{55} \\
Ring detection \cite{ref:jecmen2024}          & \ding{55} & \ding{55} & \ding{55} & \ding{55} & \ding{55} \\
Text-match attack \cite{ref:hsieh2025}     & \ding{55} & \ding{55} & \ding{55} & \ding{55} & \ding{55} \\
Cycle-free review \cite{ref:cyclefree}     & \ding{55} & \ding{55} & \ding{55} & \ding{55} & \ding{55} \\
Peer prediction \cite{ref:peerprediction}     & \ding{55} & \ding{55} & \checkmark & \ding{55} & \ding{55} \\
\midrule
This paper & \checkmark & \checkmark & \checkmark & \checkmark & \checkmark \\
\bottomrule
\end{tabular}
\end{table}

\subsection{Game theory in peer review}

Mechanism design for peer review has concentrated on eliciting honest reports
from reviewers with no interest in any specific outcome, using peer prediction
and related scoring rules~\cite{ref:peerprediction,ref:mechdesign}, and on
strategic behavior where reviewers are themselves
authors~\cite{ref:strategic-review}. Repeated-game models of reviewer behavior
exist~\cite{ref:repeated-review}, but model the reviewer's relationship with the
venue rather than with another reviewer. Our analysis rests on standard
repeated-game machinery~\cite{ref:fudenberg-tirole}, including the observation
that finite horizons need not unravel once types are uncertain. To our knowledge
no prior model treats the
formation of a collusive relationship as the object of analysis, the step open
recruitment has made cheap. Table~\ref{tab:related} states the positioning
along the axes on which this paper departs from prior work.

\subsection{Language models and review integrity}

Language models can now produce review text that human readers find
useful~\cite{ref:llm-review-quality}, and submissions carrying hidden prompts
aimed at machine reviewers have been
documented~\cite{ref:prompt-injection-papers}. Estimates place the fraction of
reviews at recent top venues showing signs of language model assistance between
15\% and 21\%~\cite{ref:llm-reviews,ref:llm-detect}, and detection
of machine-written review text is unreliable at the false positive rates a
program chair can tolerate. Our interest is more specific: writing a favorable
review that survives co-reviewers and a meta-reviewer is the one genuinely
costly step in the collusion sequence, and the step language models make cheap.
The author-drafted variant of \S\ref{sec:taxonomy} is harder still, and beyond
the reach of any machine-text detector.

%% file: sections/03_model.tex
\section{The Mutual Bidding Dilemma}
\label{sec:model}

\subsection{Setting and notation}

A conference receives $N$ submissions $\Pap$ and employs $M$ reviewers $\Rev$.
Reviewer $i$ authors the set $\Pap_i \subset \Pap$ and is assigned at most $L_i$
papers; each paper receives $k$ reviews. The affinity matrix
$\Aff \in [0,1]^{M \times N}$ comes from the venue's text-matching system, the
bid matrix $\Bid \in \mathcal{V}^{M \times N}$ records preferences over a finite
vocabulary $\mathcal{V}$, and the assignment $\Asg = f(\Aff,\Bid)$ solves
\eqref{eq:assignment}. Reviewer $i$ then reports a score
$s_{ij} \in [\underline{s},\overline{s}]$ for each assigned paper $j$, and a
paper is accepted if its aggregate score exceeds a threshold set to meet the
venue's target acceptance rate. Table~\ref{tab:notation} collects the game
parameters, all per-dyad and per-round unless stated otherwise.

\begin{definition}[Collusive dyad]
\label{def:dyad}
A pair of reviewers $(i,j)$ forms a \emph{collusive dyad} if they satisfy no
declared conflict of interest, have no co-authorship or institutional
relationship, have not interacted before the current recruitment episode, and
agree to bid on each other's submissions with the intention of returning
inflated scores conditional on assignment.
\end{definition}

The third clause separates our object of study from the collusion modeled in
prior work, and makes the analysis non-trivial by removing the trust those
models assume.

\begin{table}[htbp]
\centering
\caption{Game parameters.}
\label{tab:notation}
\footnotesize
\setlength{\tabcolsep}{3pt}
\begin{tabular}{@{}l p{0.79\columnwidth}@{}}
\toprule
Symbol & Meaning \\
\midrule
$b$        & value to a colluder of one inflated review on their own paper \\
$\rho$     & $\Prob(\text{assigned to target paper} \mid \text{collusive bid})$ \\
$\rho_0$   & $\Prob(\text{assigned to target paper} \mid \text{no bid})$ \\
$d$        & per-inflated-review probability of detection \\
$\pi$      & sanction if detected \\
$\kappa$   & $= d\pi$, expected sanction cost of one inflated review \\
$c_e$      & extra effort to write a convincing inflated review \\
$c_b$      & opportunity cost of spending a bid slot collusively \\
$c_s$      & cost of posting or answering a recruitment signal \\
$\mu$      & probability a public recruitment post is observed by the venue \\
$\pi_s$    & sanction if a recruitment post is attributed to its author \\
$q$        & $\Prob(\text{responder is a genuine colluder})$ \\
$L$        & loss to the first mover if reported at the exchange stage \\
$R$        & payoff to a responder from reporting rather than reciprocating \\
$\delta_r$ & discount factor between consecutive submission deadlines \\
$m$        & number of partners a reviewer maintains in one round \\
\bottomrule
\end{tabular}
\end{table}

\subsection{The four-stage game}

The dyad plays a four-stage extensive-form game. Stages are ordered by the
conference calendar, and each is played under a distinct information structure,
which is the substance of the model. Nature moves first, at Stage~1, to draw the
responder's type. Bids are private, so Stage~3 is played under a non-trivial
information set, and the Stage~4 subgame is reached only on the assignment
event, which has probability $\rho^2$ under independent assignment.

\textbf{Stage 1: recruitment signaling.} The initiator chooses whether to post a
recruitment signal reaching $n$ potential partners, at cost $c_s$. The post is
itself evidence: the venue observes and attributes it with probability $\mu$,
carrying sanction $\pi_s$. Each reader responds independently with probability
$\lambda$, so the initiator meets at least one responder with probability
$1-(1-\lambda)^n$. Nature draws the responder's type: with probability $q$ a
genuine colluder, and with probability $1-q$ someone who intends to gather
evidence and report, which we call a sting. The post is a costly signal, and the
initiator cannot distinguish the types at this stage.

\textbf{Stage 2: information exchange.} Targeted bidding requires knowing which
submission to bid on, and simultaneous exchange is unavailable to strangers
without a commitment device, so exchange is sequential. The first mover chooses
from $\{\texttt{reveal}, \texttt{withdraw}\}$; the responder, having received a
real identifier, chooses from $\{\texttt{reciprocate}, \texttt{fake},
\texttt{report}\}$. This is a trust game with the roles reversed, since what the
first mover hands over is not money but
evidence: revealing transmits identity, authorship, and documentary proof of
intent at once. A responder who fakes obtains nothing directly, because without
a real identifier the first mover cannot reciprocate; one who reports obtains
$R$, whether an explicit reward, standing with the venue, or the competitive
benefit of removing a rival submission. The first mover who is reported loses
$L$.

\textbf{Stage 3: bidding.} Each party independently records a bid on the
partner's paper, choosing from $\{\texttt{bid}, \texttt{no bid}\}$ at
opportunity cost $c_b$ for the former. Two features distinguish this stage from
a standard cooperation problem. The action is hidden, since bids are visible
only to the venue, and it is not decisive, since the optimizer takes bids as one
input among several and bidding yields assignment only with probability $\rho$
against a baseline $\rho_0 \ll \rho$. The combination is a moral hazard problem
with a confounded signal: a partner
who fails to be assigned may never have bid, or may have bid and lost, an event
of probability $1-\rho$ even under full cooperation, and
\S\ref{sec:equilibrium} shows this confound is the binding constraint on
enforcement.

\textbf{Stage 4: reviewing.} Conditional on assignment, each party chooses
$\actI$ (report a score above their honest assessment by a margin sufficient to
change the decision) or $\actH$ (report honestly). An inflated review that will
survive co-reviewer discussion and meta-review costs $c_e$ beyond the effort the
reviewer owes anyway, and carries expected sanction $\kappa = d\pi$. Restricting
attention to mutual assignment gives the payoffs of Table~\ref{tab:pd}.

\begin{table}[htbp]
\centering
\caption{Stage~4 payoffs conditional on mutual assignment. Row player is
reviewer $i$, column player is reviewer $j$. $\actI$ denotes an inflated score
and $\actH$ an honest one.}
\label{tab:pd}
\begin{tabular}{c|cc}
\toprule
 & $\actI$ & $\actH$ \\
\midrule
$\actI$ & $(b-\kappa-c_e,\; b-\kappa-c_e)$ & $(-\kappa-c_e,\; b)$ \\
$\actH$ & $(b,\; -\kappa-c_e)$             & $(0,\; 0)$ \\
\bottomrule
\end{tabular}
\end{table}

\begin{lemma}[Stage 4 is a Prisoner's Dilemma]
\label{lem:pd}
For all $b>0$ and $\kappa+c_e>0$, the Stage~4 game of Table~\ref{tab:pd} is a
Prisoner's Dilemma: $\actH$ strictly dominates $\actI$ for both players, mutual
$\actH$ is the unique Nash equilibrium, and mutual $\actI$ Pareto-dominates it
whenever $b > \kappa + c_e$.
\end{lemma}

\begin{proof}[Proof sketch]
Fix $j$'s action. If $j$ plays $\actI$, player $i$ compares $b-\kappa-c_e$
against $b$; if $j$ plays $\actH$, player $i$ compares $-\kappa-c_e$ against
$0$. In both cases the difference is $-(\kappa+c_e)<0$, so $\actH$ strictly
dominates, independently of $b$. The remaining claims follow. Full argument in
Appendix~\ref{app:proofs}.
\end{proof}

The independence from $b$ is the substantive content: $b$ is delivered by the
\emph{partner's} action and therefore cancels from the comparison governing
one's own, so no acceptance benefit, however large, makes reciprocation
individually rational.

\subsection{Payoffs}

Collecting the stages, the expected payoff to a party who signals, exchanges
identifiers with a genuine partner, bids, and expects reciprocation with
probability $\sigma$ is
\begin{equation}
\label{eq:utility}
U = \underbrace{\rho\,\sigma\, b}_{\text{expected gain}}
  \;-\; \underbrace{\rho\left(\kappa + c_e\right)}_{\text{cost of reciprocating}}
  \;-\; \underbrace{c_b}_{\text{bid slot}}
  \;-\; \underbrace{c_s + \mu\pi_s}_{\text{signaling}} .
\end{equation}
Under mutual cooperation, $\sigma=1$ and we write the per-round gain and cost as
\begin{equation}
\label{eq:gw}
g = \rho b, \qquad w = \rho(\kappa+c_e) + c_b,
\end{equation}
so the per-round net is $u_C = g - w$ and collusion is worth entering only if
$u_C > c_s + \mu\pi_s$.

Two parameters in \eqref{eq:gw} carry most of the paper's message. The first is
$\rho$, which the attack strategies of \S\ref{sec:taxonomy} raise and the
defenses of \S\ref{sec:experiments} lower. The second is $c_e$. A reviewer
writing a convincing favorable review by hand spends hours, because it must
withstand a discussion phase in which co-reviewers who read the paper honestly
will disagree. One using a language model spends minutes, and one who receives a
draft written by the paper's own authors spends none, with text more specific
and more technically accurate than either alternative.
\S\ref{sec:equilibrium} shows $c_e$ enters the sustainability threshold
directly.

%% file: sections/04_equilibrium.tex
\section{Equilibrium Analysis}
\label{sec:equilibrium}

\subsection{The one-shot game unravels}

Consider first the benchmark matching the phenomenon most directly: two
strangers, one conference, and no observation of each other after the reviews
are submitted.

\begin{theorem}[Unraveling]
\label{thm:unravel}
In the one-shot Mutual Bidding Dilemma with $c_s > 0$, $c_b > 0$ and
$\kappa + c_e > 0$, the unique subgame-perfect equilibrium is: no recruitment
signal is posted, no submission identifier is revealed, no collusive bid is
placed, and both parties report honestly. Collusion does not occur for any
values of $b$, $\rho$, $q$, or $\pi$.
\end{theorem}

\begin{proof}[Proof sketch]
Backward induction. Lemma~\ref{lem:pd} makes $\actH$ strictly dominant at
Stage~4; a party anticipating an honest partner then obtains $-c_b < 0$ from
bidding at Stage~3; anticipating no bid, revealing at Stage~2 yields
$-(1-q)L \le 0$ against $0$; and at Stage~1 the continuation is at most $0$
against a cost $c_s + \mu\pi_s > 0$. Each step is strict, so the equilibrium is
unique. Appendix~\ref{app:proofs}.
\end{proof}

Theorem~\ref{thm:unravel} is the pivot of the paper. Collusion of exactly the
kind in Fig.~\ref{fig:screenshots} is observed at increasing volume, and the
one-shot model says it should never occur. The resolution cannot lie in payoff
magnitudes, since Lemma~\ref{lem:pd} holds for every magnitude, so it must lie
in the two assumptions of the benchmark: that the interaction happens once and
that nothing is observed afterwards. Any explanation of stranger collusion
therefore requires an enforcement technology, a mechanism by which failure to
reciprocate imposes a cost outside the current stage game.

\subsection{Where the enforcement comes from}

Three features of the modern review ecosystem supply it. First, \emph{reviews of
one's own paper are always observed}, since every venue returns reviews to
authors, so a colluder learns with certainty whether an assigned partner
delivered and Stage~4 is perfectly monitored ex post. This is not true of
Stage~3, and the asymmetry drives Proposition~\ref{prop:tolerance}. Second,
\emph{the conference calendar supplies repetition}: major venues have deadlines
two to three months apart, so $\delta_r$ is dominated by the probability the
relationship continues rather than by impatience. Third, \emph{the platform
supplies reputation}. Recruitment groups persist across cycles, members
accumulate standing, and someone who takes a partner's bid and returns an honest
score can be named inside the group. The platform performs for strangers the
function institutional proximity performs for colleagues, at near-zero cost.
Open recruitment is therefore not merely a cheaper way to find partners but a
change in the structure of the game.
A fourth feature, the revision window the discussion phase opens, affects the
timing of $\kappa$ rather than its magnitude and we do not model it.

\subsection{Sustainability of reciprocal inflation}

Index conference rounds by $t=1,2,\dots$ and let $\delta_r \in (0,1)$ be the
product of time discounting and the probability the relationship survives to the
next round. Under grim trigger, cooperate in round $t$ if and only if no
observed betrayal has occurred, an observed betrayal being a round in which the
partner was assigned to one's own paper and reported honestly.

\begin{theorem}[Sustainability]
\label{thm:sustain}
Mutual cooperation is a subgame-perfect equilibrium of the repeated Mutual
Bidding Dilemma under grim trigger if and only if
\begin{equation}
\label{eq:deltastar}
\delta_r \;\ge\; \deltastar \;:=\; \frac{w}{g}
\;=\; \frac{\kappa + c_e}{b} \;+\; \frac{c_b}{\rho\,b}
\;=\; \frac{d\pi + c_e}{b} + \frac{c_b}{\rho b}.
\end{equation}
\end{theorem}

\begin{proof}[Proof sketch]
Cooperating forever is worth $V_C = (g-w)/(1-\delta_r)$ with $g,w$ as in
\eqref{eq:gw}. The most profitable one-shot deviation withholds the bid and
reports honestly if nonetheless assigned, saving $w$ while still collecting the
partner's already-committed $g$ this round and $0$ thereafter, so $V_D = g$.
Then $V_C \ge V_D$ rearranges to $\delta_r g \ge w$, and the one-shot deviation
principle~\cite{ref:fudenberg-tirole} gives subgame perfection.
Appendix~\ref{app:proofs}.
\end{proof}

\begin{corollary}[Comparative statics]
\label{cor:comparative}
$\deltastar$ is increasing in $d$, $\pi$, $c_e$, and $c_b$, and decreasing in
$b$ and $\rho$.
\end{corollary}

Two of these carry the paper's practical content. Monotonicity in $c_e$ matters
because writing a convincing inflated review is the only expensive step in the
sequence and the one that has recently become cheap: $c_e \to 0$ gives
$\deltastar \to d\pi/b + c_b/(\rho b)$, so a dyad that could not sustain
collusion at human effort levels may sustain it with machine assistance and no
other change. Monotonicity in $\rho$ matters for defense design: any
intervention lowering the probability that a collusive bid produces an
assignment, as randomized assignment~\cite{ref:jecmen2020} does directly and
cycle-free constraints~\cite{ref:cyclefree} do for particular ring topologies,
raises $\deltastar$ and shrinks the set of dyads that can sustain collusion.
Proposition~\ref{prop:tolerance} shows that it does so twice over.

\subsection{Hidden bidding and the tolerance problem}

The grim trigger of Theorem~\ref{thm:sustain} punishes only observed betrayal,
and Stage~3 produces no observation: a party who quietly stops bidding is never
assigned, appears merely unlucky, and is never punished. Enforcement at Stage~3
therefore requires punishing on non-assignment, a confounded signal that occurs
with probability $1-\rho$ in every round even under full cooperation. This is
the structure Green and Porter identified in cartels observing only a noisy
public signal of each other's conduct~\cite{ref:fudenberg-tirole}, and the
resolution is the same: punishment must be triggered by a statistic that
sometimes fires on the equilibrium path, and the optimal rule trades deterrence
against those false triggers. Let the colluders adopt a
$k$-strike rule, terminating the relationship after $k$ consecutive rounds in
which the partner is not assigned to one's paper.

\begin{proposition}[Tolerance]
\label{prop:tolerance}
Under the $k$-strike rule, withholding the bid is deterred if and only if
$\delta_r^{\,k} \ge \deltastar$, that is
\begin{equation}
\label{eq:kmax}
k \;\le\; k_{\max} \;=\; \frac{\ln \deltastar}{\ln \delta_r},
\end{equation}
while holding the per-round probability of terminating an honestly cooperating
relationship below $\varepsilon$ requires
\begin{equation}
\label{eq:kmin}
k \;\ge\; k_{\min} \;=\; \frac{\ln \varepsilon}{\ln(1-\rho)} .
\end{equation}
A workable rule exists only if $k_{\min} \le k_{\max}$. As $\rho \to 0$,
$k_{\min} \to \infty$ and $k_{\max}$ falls, so the feasible set empties.
\end{proposition}

\begin{proof}[Proof sketch]
A party who stops bidding collects $g$ rather than $g-w$ for exactly $k$ rounds
before the rule fires, so $V_{D3} = g(1-\delta_r^k)/(1-\delta_r)$ and
$V_C \ge V_{D3}$ yields $\delta_r^k g \ge w$. Equation~\eqref{eq:kmin} requires
$k$ consecutive independent non-assignments, each of probability $1-\rho$, to
occur with probability at most $\varepsilon$. Appendix~\ref{app:proofs}.
\end{proof}

Proposition~\ref{prop:tolerance} identifies an effect of randomized assignment
that has not, to our knowledge, been noted. Randomization is normally justified
as a way to reduce the expected return on manipulation and criticized for the
assignment quality it costs~\cite{ref:jecmen2020,ref:jecmen2022}. It also
degrades the colluders' own monitoring technology, since a lower $\rho$ forces a
more forgiving rule to avoid destroying honest partnerships, and a more
forgiving rule cannot deter free-riding, at no cost beyond the first effect.

\subsection{Who reveals first, and when it is worth posting}

Let $W = V_C = (g-w)/(1-\delta_r)$ be the continuation value of a genuine
partnership, available from Theorem~\ref{thm:sustain} whenever
$\delta_r \ge \deltastar$.

\begin{proposition}[Trust threshold]
\label{prop:trust}
At Stage~2, the responder reciprocates rather than reports if and only if
$W \ge R$. Anticipating this, the first mover reveals a real submission
identifier if and only if
\begin{equation}
\label{eq:qstar}
q \;\ge\; \qstar \;=\; \frac{L}{W + L}.
\end{equation}
$\qstar$ is decreasing in $W$, and $\qstar \to 0$ as $\delta_r \to 1$.
\end{proposition}

\begin{proof}[Proof sketch]
Revealing yields $qW - (1-q)L$ against $0$ from withdrawing, which rearranges to
\eqref{eq:qstar}; the responder compares the continuation value $W$ against the
one-time payoff $R$. Appendix~\ref{app:proofs}.
\end{proof}

Equation~\eqref{eq:qstar} accounts for the behavior hardest to reconcile with
intuition, namely people posting evidence of intent under identifiable handles
and then sending submission identifiers to strangers. They are not misjudging
the probability of a sting but responding correctly to a continuation value
large enough that $\qstar$ is low. The lever a venue controls here is $L$, which
saturates because it appears in both numerator and denominator; the effective
lever is $W$.

\begin{proposition}[Platform size]
\label{prop:platform}
An initiator whose post reaches $n$ readers, each responding independently with
probability $\lambda$, posts if and only if
$n \ge n^{\star} = \ln(1 - (c_s + \mu\pi_s)/\Gamma)/\ln(1-\lambda)$ with
$\Gamma = qW-(1-q)L > c_s + \mu \pi_s$. $n^{\star}$ is increasing in the venue's
monitoring intensity $\mu$ and in the sanction $\pi_s$ attached to an attributed
post.
\end{proposition}

A public recruitment group satisfies $n \ge n^\star$ comfortably at any
plausible $\lambda$, which is why the posts are public rather than private, and
raising $\mu$ is the only intervention in the model acting before any bid is
placed.

\subsection{Ring size}

A reviewer may maintain $m$ partners in one round. Detection risk grows with
$m$, since a larger collusive subgraph is a stronger signal to any of the
detectors in \S\ref{sec:taxonomy}; we write $d(m) = d_0 + \beta m$ for the
per-review detection probability.

\begin{proposition}[Optimal ring size]
\label{prop:ringsize}
The per-round payoff $\Pi(m) = m\left[\rho(b-c_e)-c_b-\rho\pi d_0\right]
- \rho\pi\beta m^2$ is maximized at
\begin{equation}
\label{eq:mstar}
\mstar \;=\; \frac{\rho(b - c_e) - c_b - \rho\pi d_0}{2\rho\pi\beta},
\qquad
\frac{\partial \mstar}{\partial c_e} = -\frac{1}{2\pi\beta} < 0 .
\end{equation}
\end{proposition}

The derivative in \eqref{eq:mstar} is the formal version of the claim that
language models enlarge collusion rings. Effort cost per partner is what bounded
ring size when reviews were written by hand, and removing it moves the optimum
outward at a rate set only by the sanction and by $\beta$. Since $\beta$ is in
the denominator, a detector sharply more sensitive to larger rings limits ring
size even when it cannot catch dyads.

\subsection{The equilibrium-breaking threshold}

\begin{theorem}[Detection threshold]
\label{thm:dstar}
For a fixed discount factor $\delta_r$, mutual cooperation is sustainable if and
only if the per-review detection probability satisfies
\begin{equation}
\label{eq:dstar}
d \;\le\; \dstar(\delta_r) \;=\; \frac{1}{\pi}
\left( \delta_r b - c_e - \frac{c_b}{\rho} \right).
\end{equation}
In particular, no discount factor sustains collusion once
$d > (b - c_e - c_b/\rho)/\pi$.
\end{theorem}

\begin{proof}[Proof sketch]
Substitute $\kappa = d\pi$ into \eqref{eq:deltastar} and solve
$\deltastar \le \delta_r$ for $d$. The unconditional statement sets
$\delta_r = 1$. Appendix~\ref{app:proofs}.
\end{proof}

Fig.~\ref{fig:phase} plots \eqref{eq:dstar} in the $(d,b)$ plane. The boundary
is linear in $b$ with slope $\delta_r/\pi$, so the benefit of publication and
the probability of detection trade off at a rate set entirely by the sanction
and the durability of the relationship. The reported failure of existing
detectors on camouflaged bidding~\cite{ref:jecmen2024}, which
\S\ref{sec:experiments} reproduces, places current venues far inside the stable
region; the annotated estimate of that position depends on parameters no venue
publishes and is an order-of-magnitude claim rather than a measurement.
\S\ref{sec:discussion} reads the available interventions off
\eqref{eq:dstar}.

\begin{figure}[htbp]
  \centering
  \includegraphics[width=\columnwidth]{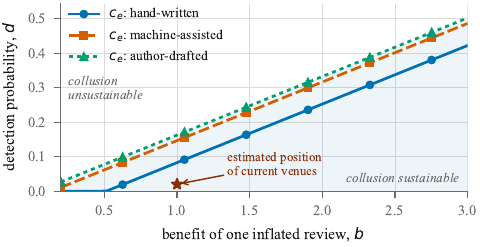}
  \caption{Stability regions in the $(d,b)$ plane from \eqref{eq:dstar}, at
  $c_e$ corresponding to hand-written and to machine-assisted reviews. Below
  each boundary, reciprocal inflation is sustainable for dyads with
  $\delta_r \ge \deltastar$. The reduction in $c_e$ moves the boundary upward, enlarging the stable region without any change to
  detection capability. Parameters as in Table~\ref{tab:params}.}
  \label{fig:phase}
\end{figure}

%% file: sections/05_taxonomy.tex
\section{Attack Strategy Taxonomy}
\label{sec:taxonomy}

The equilibrium analysis treats $\rho$ and $d$ as parameters, but a colluder
chooses them, within limits, by choosing how to bid and how to write. We catalog
those choices because defenses cannot be evaluated against an unspecified
adversary, and because the existing detector
evaluation~\cite{ref:jecmen2024} considers a narrower strategy space than
recruitment posts describe.

\subsection{Bidding strategies}

\textbf{Naive.} Bid \emph{eager} on every partner's paper and behave normally
otherwise. This maximizes $\rho$ per partner and requires no planning, but
produces the signature every published detector is built to find: positive bids
concentrated on a small set of papers with no topical justification. It is the
lower bound on adversary sophistication.

\textbf{Camouflaged.} Place the collusive bids inside a larger set of genuine
bids on topically adjacent papers, sized so the collusive fraction of the
reviewer's bid vector matches the population distribution. The collusive bids
become individually unremarkable and the aggregate profile statistically
ordinary. This is the strategy \cite{ref:jecmen2024} found sufficient to defeat
most off-the-shelf detectors, and it costs only the time spent selecting cover
papers.

\textbf{Distributed ring.} Replace the dyad $i \leftrightarrow j$ with a
directed cycle $i \to j \to \ell \to i$ of length $k \ge 3$, in which each member
bids on exactly one other member's paper. The assignment graph then contains no
two-cycle, so cycle-free constraints at $k=2$ do not bind~\cite{ref:cyclefree},
and each individual's collusive bid count is one, below the threshold of any
frequency-based test. The cost is coordination, and
Proposition~\ref{prop:tolerance} now applies to a longer monitoring chain in
which $i$ observes only $\ell$'s review and must infer $j$'s conduct
indirectly.

\textbf{Affinity-aware.} Combine a collusive bid with manipulation of the
text-matching input, by curating the reviewer's publication record toward the
partner's topic or inserting matched background text into the
submission~\cite{ref:hsieh2025}. Because the optimizer weighs affinity and bids
together, raising $A_{ij}$ raises $\rho$ for a fixed bid and, more importantly,
makes the assignment defensible: a reviewer assigned to a paper they bid on
\emph{and} match strongly is not anomalous under any signal we know of. By
Theorem~\ref{thm:sustain} it also lowers $\deltastar$, making collusion
sustainable for dyads that could not otherwise sustain it.

\begin{figure*}[htbp]
  \centering
  \begin{subfigure}[t]{0.32\textwidth}
    \includegraphics[width=\textwidth]{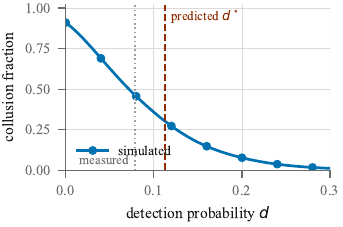}
    \caption{Detection probability $d$.}
    \label{fig:eq-d}
  \end{subfigure}\hfill
  \begin{subfigure}[t]{0.32\textwidth}
    \includegraphics[width=\textwidth]{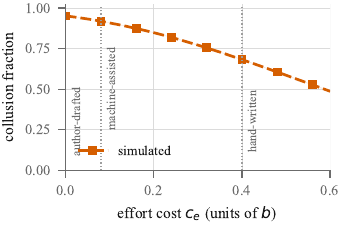}
    \caption{Effort cost $c_e$.}
    \label{fig:eq-ce}
  \end{subfigure}\hfill
  \begin{subfigure}[t]{0.32\textwidth}
    \includegraphics[width=\textwidth]{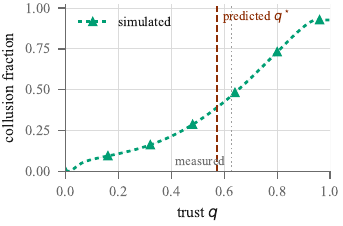}
    \caption{Trust parameter $q$.}
    \label{fig:eq-q}
  \end{subfigure}
  \caption{Equilibrium validation over $\nAgents$ heterogeneous agents. In
  (\subref{fig:eq-d}) and (\subref{fig:eq-q}) the dashed line is the closed-form
  threshold and the dotted line the measured transition; each measured
  transition sits on the congestion-adjusted side of its closed form. Panel
  (\subref{fig:eq-ce}) marks the three effort regimes of \S\ref{sec:taxonomy}.
  Shaded bands are 95\% intervals over 10 seeds.}
  \label{fig:equilibrium}
\end{figure*}

\subsection{Reviewing strategies}

The reviewing stage admits its own gradient, ordered by how much of the
meta-reviewer's attention the inflated review attracts. \emph{Uniform maximum}
awards the highest available score with brief justification, effective only if
the other reviews are also positive. \emph{Calibrated inflation} awards a score
high relative to the paper but ordinary relative to the reviewer's own history,
bounding $\Delta s$ by the reviewer's observed variance rather than by the score
scale, so a per-reviewer distribution test finds nothing. \emph{Constructive
positive} raises several minor concerns, addresses them approvingly, and
recommends acceptance, reading as engaged and surviving discussion at
substantially higher $c_e$. \emph{Author-drafted} has the submitting authors
write the review and the colluding reviewer submit it after superficial editing,
giving text that is specific, technically correct, and human written. It attains
the lowest $c_e$ here together with the highest resistance to content-based
scrutiny, and machine-text detectors do not apply because there is no machine
text. We flag it as the most concerning entry in the taxonomy.

\subsection{Detector families and evasion}

We consider four detector families spanning the published approaches.
\emph{Bid frequency} tests flag reviewers whose bid vectors are unusual in
concentration or in the affinity of the papers they bid on. \emph{Dense subgraph
mining} searches the reviewer-paper bid graph for subgraphs whose internal
density exceeds what the topic structure explains, following the peeling
approach of~\cite{ref:charikar} and its camouflage-resistant
refinements~\cite{ref:fraudar}. \emph{Cycle detection} searches the assignment or
bid graph for short directed cycles among authors who are also reviewers.
\emph{Text-similarity anomaly} tests flag assignments in which realized affinity
is implausible given the reviewer's publication history, and reviews whose
content diverges from the reviewer's usual output.

Each family is defeated by a specific and cheap countermeasure, and no family
covers the countermeasures to the others. Camouflage defeats frequency and
density tests by making the collusive bids a small fraction of an ordinary
profile, distribution defeats cycle detection at the length venues constrain,
and affinity awareness defeats similarity tests by making the match genuine as
measured. Under our simulation, a colluder combining all three is not
distinguishable from an honest reviewer by any of the four signals taken alone.
That is a claim about single-signal detectors, not a claim that such colluders
are undetectable, and \S\ref{sec:discussion} argues the joint distribution of
the signals remains informative. Table~\ref{tab:evasion} in
Appendix~\ref{app:detection} gives the predicted outcome of each pairing, which
\S\ref{sec:experiments} measures.

%% file: sections/06_experiments.tex
\section{Experimental Evaluation}
\label{sec:experiments}

\subsection{Simulation framework}

No conference releases bidding data, assignment logs, or affinity scores, so the
quantities this paper needs cannot be measured on any existing dataset. We build
an end-to-end simulator and calibrate it against what is public: submission and
reviewer counts, review load, score distributions, and acceptance rates from
OpenReview~\cite{ref:openreview}, and bidding patterns from the simulated
malicious bidding dataset of~\cite{ref:jecmen2023}, the same evidentiary posture
as prior work in this area.

The simulator runs the full pipeline: it generates submissions with a latent
quality $\theta_p$ and a topic vector and reviewers with topic vectors and
authorship, computes affinity from topic overlap with noise, produces honest
bids as a function of affinity, injects collusive bids according to a chosen
strategy from \S\ref{sec:taxonomy}, solves \eqref{eq:assignment}, scores papers
as noisy observations of $\theta_p$ with collusive scores inflated, and applies
an acceptance threshold calibrated to the target rate. Because $\theta_p$ is
generated rather than inferred, the simulator supplies the ground-truth quality
no real dataset contains, which is what makes the measurements of
\S\ref{sec:impact} possible at all. Venue parameters follow published ICLR and
NeurIPS statistics at $N = 6000$, $M = 9000$, $k=3$ reviews per paper and load
$L_i = 5$; game parameters are swept over the ranges in
Table~\ref{tab:params}, which also records the assumptions we could not
calibrate. The framework and experiment configurations are available on
request.

\subsection{Experiment 1: equilibrium validation}
\label{sec:e1}

\textbf{Setup.} We populate the game with $\nAgents$ agents, heterogeneous in
$b$ and $\delta_r$, who decide whether to collude by best response to the
current population state, and iterate to a fixed point. Two couplings make this
a population game rather than a threshold applied pointwise: more colluders
compete for the same assignment slots, lowering $\rho$, and a larger collusive
subgraph is easier to detect, raising $d$. We sweep $d$, $c_e$, and $q$
independently and record the collusion fraction at the fixed point.

\begin{figure*}[htbp]
  \centering
  \begin{subfigure}[t]{0.32\textwidth}
    \includegraphics[width=\textwidth]{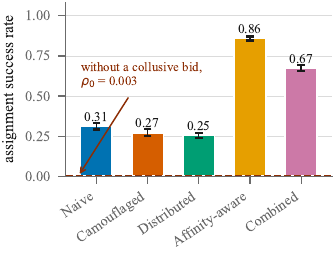}
    \caption{Assignment success by strategy.}
    \label{fig:success}
  \end{subfigure}\hfill
  \begin{subfigure}[t]{0.32\textwidth}
    \includegraphics[width=\textwidth]{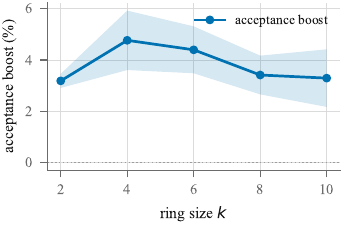}
    \caption{Acceptance gain against ring size.}
    \label{fig:impact-ring}
  \end{subfigure}\hfill
  \begin{subfigure}[t]{0.32\textwidth}
    \includegraphics[width=\textwidth]{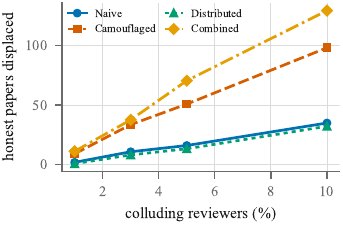}
    \caption{Displacement against colluder fraction.}
    \label{fig:impact-disp}
  \end{subfigure}
  \caption{Attack effectiveness and downstream impact, each measured against the
  zero-collusion run on the same instance. In (\subref{fig:success}) the dashed
  line is $\rho_0$, the probability of the same assignment without a collusive
  bid, over $\evasionSeeds$ instances. Panels (\subref{fig:impact-ring}) and
  (\subref{fig:impact-disp}) use $\impactSeeds$ instances; the peak at
  $k = \bestRingSize$ is the interior optimum of
  Proposition~\ref{prop:ringsize}.}
  \label{fig:impact}
\end{figure*}

\textbf{Results.} The transition in $d$ (Fig.~\ref{fig:eq-d}) occurs at
$d = \dTransition$ against the closed-form $\dstar = \dStar$ from
\eqref{eq:dstar}. It sits below the predicted threshold, and the gap is the
congestion term: at a positive collusion fraction the effective detection
probability exceeds its nominal value, so the population crosses the threshold
before the nominal parameter reaches it. The same mechanism explains
Fig.~\ref{fig:eq-q}, where the measured transition $q = \qTransition$ sits above
the closed-form $\qstar = \qStar$, since congestion lowers $W$ and
\eqref{eq:qstar} is decreasing in $W$. The predicted $\qstar$ is well below one,
the formal counterpart of people sending submission identifiers to strangers who
may be hostile.

The $c_e$ sweep of Fig.~\ref{fig:eq-ce} is Corollary~\ref{cor:comparative}
measured. At hand-written effort ($c_e = \effortCostHand\,b$),
$\collusionHand$\% of the population can sustain collusion; at machine-assisted
effort ($c_e = \effortCostMachine\,b$) it is $\collusionMachine$\%; and at the
author-drafted level $c_e = 0$ it is $\collusionAuthor$\%. Nothing about
detection, sanctions, or the value of publication changed across those three
points. Over the same range the optimal ring size of
Proposition~\ref{prop:ringsize} moves from $\mstar = \ringSizeHand$ to
$\mstar = \ringSizeAuthor$.

\subsection{Experiment 2: detection evasion}
\label{sec:e2}

\textbf{Setup.} We inject colluders at $\colluderFrac$\% of the reviewer pool
using each bidding strategy of \S\ref{sec:taxonomy}, and run the four detector
families over the resulting bid and assignment data. Every detector is
thresholded at a common $\evasionFPR$\% false positive rate, since a program
chair's binding constraint is how many reviewers they can afford to investigate.
We report $F_1$ over colluder identification and the assignment success rate,
the fraction of collusive bids producing the intended assignment, over
$\evasionSeeds$ instances.

\textbf{Results.} In Table~\ref{tab:detection} no detector family exceeds
$F_1 = \fOneBestOverall$ against any strategy. Cycle detection is the only
signal that identifies dyads at all ($F_1 = \fOneNaiveCycle$ against naive
bidding), and the distributed ring reduces it to $\fOneDistCycle$ by adding one
member. The bid-frequency test scores $\fOneNaiveFreq$ against every dyadic
strategy: once partners are screened for topical proximity, a collusive bid sits
inside the reviewer's normal bidding range and nothing anomalous remains to
measure.

Fig.~\ref{fig:success} gives the other half. A collusive bid raises the
probability of landing on the target paper from $\rho_0 = \rhoZero$ to
$\succNaive$ for naive bidding and to $\succAff$ when affinity manipulation is
added, a factor of over two hundred in the latter case, and the combined
strategy reaches $\succComb$ while holding every detector at or below
$\fOneCombWorst$. Evasion and effectiveness are usually in tension and here they
are not, because affinity awareness raises $\rho$ and lowers detectability for
the same reason: it makes the assignment genuinely well matched as the system
measures it.

One result runs against the taxonomy's prediction. Camouflage slightly
\emph{raises} the cycle-detection score, from $\fOneNaiveCycle$ to
$\fOneCamoCycle$, because a reviewer who bids more widely has more chances to
form an incidental two-cycle with an honest author. That the countermeasures
pull against each other is the strongest available argument for multi-signal
detection.

\begin{table}[htbp]
\centering
\caption{Detection $F_1$ at a common $\evasionFPR$\% false positive rate, mean
over $\evasionSeeds$ instances. Lower is better for the attacker; the best
(lowest) entry in each column is bold.}
\label{tab:detection}
\small
\setlength{\tabcolsep}{4pt}
\begin{tabular}{lcccc}
\toprule
Strategy & Bid freq. & Dense subgr. & Cycle det. & Text anom. \\
\midrule
Naive          & \cellNaiveFreq & \cellNaiveDense & \cellNaiveCycle & \cellNaiveText \\
Camouflaged    & \cellCamoFreq  & \cellCamoDense  & \cellCamoCycle  & \cellCamoText  \\
Distributed    & \cellDistFreq  & \cellDistDense  & \cellDistCycle  & \cellDistText  \\
Affinity-aware & \cellAffFreq   & \cellAffDense   & \cellAffCycle   & \cellAffText   \\
Combined       & \cellCombFreq  & \cellCombDense  & \cellCombCycle  & \cellCombText  \\
\bottomrule
\end{tabular}
\end{table}

\subsection{Experiment 3: downstream impact}
\label{sec:impact}

\textbf{Setup.} We hold the conference instance fixed and vary the colluder
fraction and ring size, comparing against a zero-collusion run on the same seed.
Honest reviewers score honest papers identically across the two arms by
construction, so every difference is attributable to the collusion. We measure
the acceptance rate of colluders' papers, the honest papers displaced, and the
ground-truth quality on both sides of each swap.

\begin{figure*}[htbp]
  \centering
  \begin{subfigure}[t]{0.48\textwidth}
    \includegraphics[width=\textwidth]{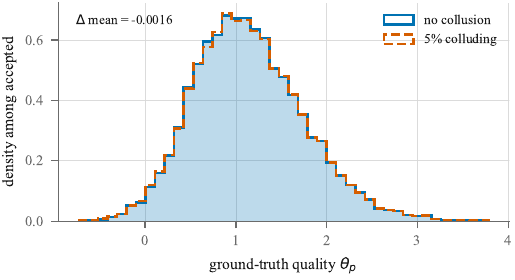}
    \caption{The accepted set in aggregate.}
    \label{fig:qual-agg}
  \end{subfigure}\hfill
  \begin{subfigure}[t]{0.48\textwidth}
    \includegraphics[width=\textwidth]{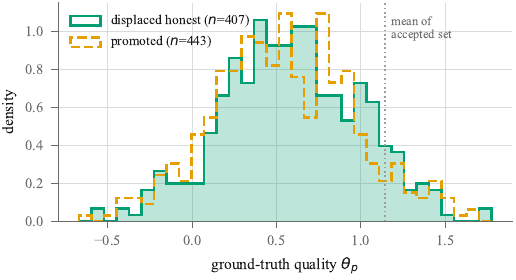}
    \caption{The papers that changed hands.}
    \label{fig:qual-swap}
  \end{subfigure}
  \caption{Ground-truth quality with and without collusion at a
  $\colluderFrac$\% colluder fraction. In (\subref{fig:qual-agg}) the two
  distributions are almost identical and mean quality moves by $\qualityDrop$.
  In (\subref{fig:qual-swap}) the displaced honest papers and the papers
  promoted in their place occupy the same narrow band at the acceptance
  threshold. The attack is invisible to any aggregate audit precisely because it
  operates where the two populations overlap.}
  \label{fig:quality}
\end{figure*}

\textbf{Results.} A bare dyad, the smallest arrangement possible and the one the
recruitment posts describe, raises the acceptance probability of a participant's
paper from $\acceptBase$\% to $\acceptDyad$\%, a gain of $\boostDyad$ percentage
points (Fig.~\ref{fig:impact-ring}). This is the quantity an individual reviewer
weighs when deciding whether to answer a recruitment post, and it requires no
organization beyond one exchange of messages. The gain peaks at
$\boostBestRing$ points at ring size $\bestRingSize$ and then falls, the
interior optimum of Proposition~\ref{prop:ringsize} appearing in the simulation:
past a point, additional partners compete for the same three review slots and
add detection exposure without adding inflated reviews. At a $\colluderFrac$\%
colluder fraction the combined strategy pushes $\dispCombFive$ honest papers out
of the accepted set and camouflaged dyads $\dispCamoFive$; at $10$\% the
combined figure is $\dispCombTen$ (Fig.~\ref{fig:impact-disp}). Distributed
rings displace fewer ($\dispDistFive$) because each member receives one collusive
review rather than $k-1$, the price of evading cycle detection.

\textbf{The harm is distributional, not aggregate.} Fig.~\ref{fig:qual-agg}
shows what a venue would observe, and the answer is essentially nothing. Mean
ground-truth quality of the accepted set falls by $\qualityDrop$ on a
standard-normal scale and the tenth percentile by $\qualityPTenDrop$, leaving
the two distributions visually indistinguishable even though tens of papers
changed hands. We expected a thickened lower tail and did not find one: both
sides of every swap sit near the acceptance threshold, so a displaced honest
paper and the promoted paper replacing it have almost the same latent quality,
which Fig.~\ref{fig:qual-swap} shows directly. This is not reassurance. No
aggregate audit will detect the attack, a venue can be substantially manipulated
while every summary statistic it publishes stays clean, and the entire cost
falls on the identifiable authors of the $\dispCombFive$ displaced papers.
Collusion at these scales does not make the conference worse on average; it
makes it arbitrary at the margin, where most submissions live.

\subsection{Experiment 4: defense effectiveness}
\label{sec:e4}

\textbf{Setup.} We apply three defenses to the same instances: cycle-free
assignment at $k=2$~\cite{ref:cyclefree}, randomized assignment with a cap on
per-pair assignment probability~\cite{ref:jecmen2020}, and both together. For
each we measure the residual collusion success rate and the cost to honest
assignments, the latter as the mean affinity of assignments to non-colluding
reviewers relative to the unconstrained optimum.

\begin{figure*}[htbp]
  \centering
  \begin{subfigure}[t]{0.48\textwidth}
    \includegraphics[width=\textwidth]{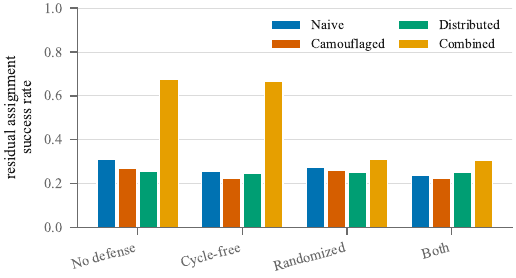}
    \caption{Residual assignment success by strategy.}
    \label{fig:def-resid}
  \end{subfigure}\hfill
  \begin{subfigure}[t]{0.48\textwidth}
    \includegraphics[width=\textwidth]{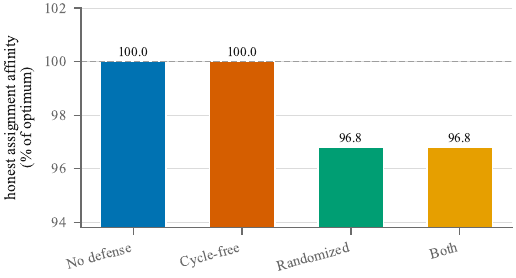}
    \caption{Cost to honest assignments.}
    \label{fig:def-cost}
  \end{subfigure}
  \caption{Defense effectiveness over $\impactSeeds$ instances. Cycle-free
  assignment is nearly free (\subref{fig:def-cost}) but only reaches dyads, and
  only those dyads whose two directions were both assigned. Randomization is the
  only defense in (\subref{fig:def-resid}) that reduces the combined strategy.}
  \label{fig:defenses}
\end{figure*}

\textbf{Results.} Cycle-free assignment at $k=2$ reduces naive dyadic success
from $\residNoneNaive$ to $\residCycleNaive$ and camouflaged dyads from
$\residNoneCamo$ to $\residCycleCamo$, leaves distributed rings untouched
($\residNoneDist$ to $\residCycleDist$), and barely moves the combined strategy
($\residNoneComb$ to $\residCycleComb$), retaining $\qualityCycle$\% of optimal
honest affinity because two-cycles are rare among honest assignments.

That the defense targeting the arrangement the posts describe is evaded by
adding one member to it is expected. The second feature is not: even against the
dyad it was designed for, cycle-free assignment removes only about a fifth of
collusive assignments. The constraint forbids a two-cycle, which requires
\emph{both} directions to be assigned, and at a realized per-direction
probability near $\rhoNone$ most dyads never reach that state. A colluder
assigned to the partner's paper while the partner is not assigned to theirs
violates no cycle constraint and is left free to submit the inflated review.
Cycle-free assignment prevents mutual favors from being consummated in the same
venue but not one-directional ones, and Theorem~\ref{thm:sustain} says
one-directional favors are exactly what a repeated relationship is built on.

Randomized assignment reduces success across every strategy rather than one
family, and does the most damage to the strongest attack: the combined strategy
falls from $\residNoneComb$ to $\residRandComb$, a larger absolute reduction
than any other cell, at $\qualityRand$\% retained honest affinity. Both
defenses together reach $\residBothComb$ at $\qualityBoth$\% retained quality.

The more consequential quantity is not in Fig.~\ref{fig:defenses} and follows
from Proposition~\ref{prop:tolerance}. Under no defense the realized assignment
probability for the combined strategy is $\rho = \rhoNone$, which by
\eqref{eq:kmin} lets colluders run a $k$-strike rule with $k \ge \kMinNone$
while \eqref{eq:kmax} permits up to $k \le \kMaxNone$, so the partnership can
police itself. Under randomization $\rho$ falls to $\rhoRand$, pushing
$k_{\min}$ to $\kMinRand$ against a $k_{\max}$ of $\kMaxRand$, and the feasible
set is empty: no rule both deters a partner from quietly withholding their bid
and avoids destroying honest partnerships through unlucky runs of
non-assignment. The residual rates in Fig.~\ref{fig:def-resid} therefore
overstate what randomization leaves standing among genuine strangers, since they
assume colluders cooperate, the assumption the mechanism undermines.

%% file: sections/07_discussion.tex
\section{Discussion}
\label{sec:discussion}

\subsection{Levers available to a program chair}

The model identifies five levers, and they are not equally available.

\emph{Detection probability $d$} is the lever the field has worked on, and
Theorem~\ref{thm:dstar} says it works once it clears $\dstar$ and does little
below it. A detector twice as good as its predecessor, while both sit well below
$\dstar = \dStar$, changes nothing about the equilibrium, so detection work
should be evaluated against the threshold rather than against relative
improvement. \S\ref{sec:e2}, where nothing exceeds
$F_1 = \fOneBestOverall$, places current venues far from it.

\emph{Assignment probability $\rho$} is the lever we would prioritize, on the
strength of Proposition~\ref{prop:tolerance} and \S\ref{sec:e4}: it enters
$\deltastar$ directly and separately destroys the colluders' ability to police
each other, the second effect free once the first is paid for. The existing
literature prices randomized assignment against its quality cost while counting
only the first~\cite{ref:jecmen2020,ref:jecmen2022}, so the mechanism is better
than its own analysis claims. It works through the colluders' beliefs about
$\rho$ as much as through $\rho$ itself, since what a partnership needs is
confidence that non-assignment signifies bad luck rather than betrayal, so a
venue adopting a per-pair cap should publish it: concealing the cap forfeits the
enforcement effect, the larger of the two.

\emph{Monitoring intensity $\mu$} at Stage~1 is the only lever acting before a
bid exists, by Proposition~\ref{prop:platform}, and methods for identifying
coordinated behavior on public platforms are well developed. A venue that states
it monitors those
channels, and that attributed posts carry a sanction, raises $n^\star$ for
everyone reading the post rather than only for those it catches; the deterrent
need not be effective, only credible.

\emph{Sanction magnitude $\pi$ and reported-first-mover loss $L$} appear in
\eqref{eq:dstar} and \eqref{eq:qstar}, and their effect saturates, since $L$
appears in both numerator and denominator of $\qstar$. \emph{Effort cost $c_e$}
is the lever most under discussion and least under control: policies restricting
machine assistance raise it for reviewers who comply, but the author-drafted
variant of \S\ref{sec:taxonomy} reaches $c_e \approx 0$ without using a machine
at all.

In the near term, then, adopt randomized assignment with an explicit per-pair
cap, publish it, and state a monitoring policy for public recruitment channels.
In the medium term, build detection that fuses bidding, scoring, assignment, and
authorship signals, since \S\ref{sec:e2} shows
each alone is evadable by a cheap countermeasure while the countermeasures are
mutually constraining. In the longer term, the $b$ term in every threshold here
is the value of a single top-venue acceptance, which no mechanism internal to
the review process can change.

\subsection{Ethics}

This paper describes attacks on a system its authors depend on. The recruitment
threads in Fig.~\ref{fig:screenshots} are public. We typeset them rather than
reproduce screenshots, because the platform is identifiable from its interface,
and we translate and lightly paraphrase so that no thread can be recovered by
searching a quoted string. Handles, avatars, timestamps, group identifiers, and
venue names are redacted, we name no individuals or groups, and we report
neither the platform nor the search terms that locate the threads. The taxonomy in \S\ref{sec:taxonomy} describes
strategies more effective than the published detector evaluations assume; we
include it because a defense cannot be evaluated against an unspecified
adversary, and because the strategies are discussed openly in the recruitment
channels and, for affinity manipulation, already published~\cite{ref:hsieh2025}.
Lowering $q$ amounts to running sting operations against a venue's own
reviewers, which we note in \eqref{eq:qstar} without recommending.
Appendix~\ref{app:ethics} gives the full statement.

\subsection{Limitations}

The results rest on simulation, because no venue releases the bidding,
assignment, or affinity data a direct measurement would need.
Table~\ref{tab:params} constrains submission and reviewer counts, load, score
distributions, and acceptance rates to published values, and leaves $b$, $\pi$,
$\mu$, $q$, and $\delta_r$ unconstrained. We sweep those rather than assert them,
and emphasize the thresholds and the directions of the comparative statics
rather than absolute magnitudes; the estimate of where current venues sit in
Fig.~\ref{fig:phase} is the weakest quantitative claim in the paper. The model
also assumes agents best-respond, which cuts in a specific direction: an
unsophisticated colluder is easier to detect than our model's, so the evasion
results characterize the upper end of adversary capability rather than the
typical case. Appendix~\ref{app:limitations} treats the fixed-pair repeated
game, the single-venue assumption, and the treatment of $d$ as exogenous against
an adversary who adapts.

%% file: sections/08_conclusion.tex
\section{Conclusion}
\label{sec:conclusion}

Reviewers recruit strangers on public platforms to bid on each other's
submissions and return inflated scores. We gave the first game-theoretic model
of how such an arrangement forms, where prior work models how an already
established group behaves, and the model says it should not form at all:
reciprocal inflation is strictly dominated at the reviewing stage for every
parameter setting, and the game unravels back to no recruitment. What sustains it is the recruitment platform and the conference
calendar acting together: reviews of one's own paper are always observed,
deadlines recur every few months, and group membership carries reputation. It
becomes sustainable exactly when the effective discount factor clears
$\deltastar = (d\pi + c_e)/b + c_b/(\rho b)$, which yields a detection threshold
$\dstar$ above which nothing sustains it. Mapping four bidding and four reviewing
strategies against four detector families then adds three findings the theory
does not supply. No detector family exceeds $F_1 = \fOneBestOverall$ once
partners are screened for topical proximity, and camouflage, ring distribution,
and affinity manipulation defeat a different family each, so no single signal
covers the countermeasures to the others. A two-person arrangement is already
worth $\boostDyad$ percentage points of acceptance probability. And the harm is
distributional rather than aggregate, with $\dispCombFive$ honest papers
displaced while mean accepted quality moves by $\qualityDrop$, so no aggregate
audit would reveal it. The defense implication is not a new detector: by
Proposition~\ref{prop:tolerance}, lowering the probability that a collusive bid
produces an assignment attacks the colluders' ability to police each other, not
merely their expected payoff, because a partner who fails to be assigned is
indistinguishable from one who never bid.

%% file: sections/09_appendix.tex
\section{Proofs}
\label{app:proofs}

Throughout, $g = \rho b$ and $w = \rho(\kappa + c_e) + c_b$ as in
\eqref{eq:gw}, with $\kappa = d\pi$.

\subsection{Proof of Lemma~\ref{lem:pd}}

Conditional on mutual assignment, the payoffs are those of Table~\ref{tab:pd}.
Fix $j$'s action. If $j$ plays $\actI$, then $i$'s payoff from $\actI$ is
$b - \kappa - c_e$ and from $\actH$ is $b$, a difference of $-(\kappa + c_e)$.
If $j$ plays $\actH$, then $i$'s payoff from $\actI$ is $-\kappa - c_e$ and from
$\actH$ is $0$, again a difference of $-(\kappa + c_e)$. Since
$\kappa + c_e > 0$, $\actH$ strictly dominates $\actI$ regardless of $j$'s
action and regardless of $b$. By symmetry the same holds for $j$, so
$(\actH,\actH)$ is the unique Nash equilibrium and it is in strictly dominant
strategies. It yields $(0,0)$, which is Pareto-dominated by
$(b-\kappa-c_e,\,b-\kappa-c_e)$ whenever $b > \kappa + c_e$. \hfill$\square$

\begin{remark}
  The independence of the dominance argument from $b$ is the substantive content.
  The benefit $b$ enters a player's payoff only through the \emph{partner's}
  action and therefore appears with the same coefficient in both branches of the
  comparison. Raising the value of a publication makes collusion more attractive
  to enter, by raising $g$, but does not make reciprocation individually rational
  once entered.
\end{remark}

\subsection{Proof of Theorem~\ref{thm:unravel}}

We solve the four stages by backward induction, establishing strictness at each
step so that the equilibrium is unique.

\emph{Stage 4.} By Lemma~\ref{lem:pd}, $\actH$ is strictly dominant at every
node in the Stage~4 subgame, so both parties report honestly whenever the
subgame is reached.

\emph{Stage 3.} Consider party $i$ choosing whether to place a collusive bid on
$P_j$. By the previous step, $j$ reports honestly if assigned, so $i$'s expected
benefit from the bid is $\rho \cdot 0 = 0$ while its cost is $c_b > 0$. Not
bidding yields $0$. Hence not bidding is strictly optimal. The same argument
applies to $j$. Note this step does not depend on $i$'s beliefs about whether
$j$ bid, so the conclusion holds at every node of $i$'s Stage~3 information set.

\emph{Stage 2.} By the previous step, no collusive bid is placed on any path, so
the continuation value of a formed partnership is $0$. The first mover's payoff
from revealing is $q \cdot 0 - (1-q)L = -(1-q)L$, and from withdrawing is $0$.
For $q < 1$ and $L > 0$ withdrawing is strictly optimal; for $q = 1$ the two are
equal and revealing is weakly optimal but yields nothing, so no collusive bid
follows in either case. The responder's choice is off the equilibrium path;
sequential rationality requires only that reporting be optimal there, which
holds whenever $R > 0$.

\emph{Stage 1.} The initiator's expected payoff from posting is at most
$\left[1-(1-\lambda)^n\right]\cdot 0 - c_s - \mu\pi_s < 0$, against $0$ from not
posting, so not posting is strictly optimal.

Each step is strict, so the profile is the unique subgame-perfect equilibrium.
No step invoked the magnitudes of $b$, $\rho$, $q$, or $\pi$, so the conclusion
holds for all of them. \hfill$\square$

\subsection{Proof of Theorem~\ref{thm:sustain}}

Consider the infinitely repeated game with per-round discount factor
$\delta_r \in (0,1)$ and the grim-trigger profile: play the collusive action
(bid, and report $\actI$ if assigned) in round $t$ if and only if no observed
betrayal has occurred in rounds $1,\dots,t-1$, where an observed betrayal is a
round in which the partner was assigned to one's own paper and reported $\actH$;
otherwise play the honest action forever.

\emph{Cooperation value.} Under mutual cooperation the per-round expected payoff
is
\[
  u_C = \underbrace{\rho b}_{g} - \underbrace{\left[\rho(\kappa+c_e) + c_b\right]}_{w},
\]
since the partner is assigned to one's paper with probability $\rho$ and inflates
when assigned, while one's own bid costs $c_b$ and produces an inflated review,
with its sanction risk and effort, with probability $\rho$. Hence
$V_C = (g-w)/(1-\delta_r)$.

\emph{Best deviation.} By the one-shot deviation principle it suffices to
consider a single-round deviation. The available deviations are: withhold the
bid; bid but report $\actH$ if assigned; or both. Withholding the bid saves
$c_b$ and, since assignment probability drops to $\rho_0$, saves
$\rho(\kappa + c_e)$ up to a term of order $\rho_0$, which we take to be zero.
Reporting $\actH$ when assigned saves $\kappa + c_e$ in that event. Doing both
saves $w$ in total, which is the maximum. In the deviation round the partner has
already committed, so the deviator still collects $g$ in expectation. Both forms
of deviation are detected: reporting $\actH$ when assigned is directly observed,
and withholding the bid is caught by the assignment-based rule of
Proposition~\ref{prop:tolerance} at $k=1$. Continuation after the trigger is
$0$, since mutual honest play yields no collusive payoff and, by
Theorem~\ref{thm:unravel} applied to the continuation, no party re-enters. Hence
$V_D = g$.

\emph{Condition.} Cooperation is sustainable iff $V_C \ge V_D$:
\[
  \frac{g - w}{1-\delta_r} \;\ge\; g
  \iff g - w \;\ge\; g(1-\delta_r)
  \iff \delta_r g \;\ge\; w ,
\]
that is $\delta_r \ge w/g = \deltastar$. Substituting \eqref{eq:gw},
\[
  \deltastar = \frac{\rho(\kappa + c_e) + c_b}{\rho b}
  = \frac{\kappa + c_e}{b} + \frac{c_b}{\rho b}
  = \frac{d\pi + c_e}{b} + \frac{c_b}{\rho b}.
\]
The profile is subgame perfect because the punishment phase is a Nash
equilibrium of the stage game by Lemma~\ref{lem:pd}. \hfill$\square$

\subsection{Proof of Corollary~\ref{cor:comparative}}

Differentiating $\deltastar = (d\pi + c_e)/b + c_b/(\rho b)$:
$\partial \deltastar/\partial d = \pi/b > 0$,
$\partial \deltastar/\partial \pi = d/b > 0$,
$\partial \deltastar/\partial c_e = 1/b > 0$,
$\partial \deltastar/\partial c_b = 1/(\rho b) > 0$,
$\partial \deltastar/\partial \rho = -c_b/(\rho^2 b) < 0$, and
$\partial \deltastar/\partial b = -\left[(d\pi + c_e) + c_b/\rho\right]/b^2 < 0$.
\hfill$\square$

\subsection{Proof of Proposition~\ref{prop:tolerance}}

\emph{Deterrence.} Under the $k$-strike rule, a party who withholds the bid is
assigned to the partner's paper with probability $\rho_0 \approx 0$, so the
strike counter increments every round and the relationship terminates after
exactly $k$ rounds. During those $k$ rounds the deviator collects $g$ per round
rather than $g - w$, since the partner continues to cooperate. Hence
\[
  V_{D3} = g\sum_{t=0}^{k-1}\delta_r^t = g\,\frac{1-\delta_r^k}{1-\delta_r},
  \qquad
  V_C = \frac{g-w}{1-\delta_r}.
\]
Requiring $V_C \ge V_{D3}$ gives $g - w \ge g(1-\delta_r^k)$, that is
$\delta_r^k g \ge w$, or $\delta_r^k \ge \deltastar$. Taking logarithms and
noting $\ln\delta_r < 0$ reverses the inequality:
$k \le \ln\deltastar/\ln\delta_r = k_{\max}$. Setting $k=1$ recovers
Theorem~\ref{thm:sustain}.

\emph{False termination.} Under mutual cooperation the partner is not assigned
in a given round with probability $1-\rho$, independently across rounds. The
rule fires on a run of $k$ consecutive such rounds, which has probability
$(1-\rho)^k$. Requiring this to be at most $\varepsilon$ gives
$k \ln(1-\rho) \le \ln\varepsilon$, and since $\ln(1-\rho)<0$,
$k \ge \ln\varepsilon/\ln(1-\rho) = k_{\min}$.

\emph{Feasibility.} A rule exists iff $k_{\min} \le k_{\max}$, that is
\[
  \frac{\ln\varepsilon}{\ln(1-\rho)} \;\le\; \frac{\ln\deltastar}{\ln\delta_r}.
\]
As $\rho \to 0^+$ we have $\ln(1-\rho) \to 0^-$, so $k_{\min}\to\infty$;
simultaneously $\deltastar \to \infty$ through the $c_b/(\rho b)$ term, so the
sustainability condition of Theorem~\ref{thm:sustain} fails outright. Both
effects act in the same direction. \hfill$\square$

\subsection{Proof of Proposition~\ref{prop:trust}}

\emph{Responder.} Having received a real identifier, the responder compares
reciprocating, which yields the continuation value $W$ of a genuine partnership,
against reporting, which yields the one-time payoff $R$. Sharing a fake
identifier yields $0$, since without a real identifier from the responder the
first mover cannot bid on the responder's paper and no partnership forms. Hence
reciprocation is optimal iff $W \ge R$.

\emph{First mover.} Anticipating reciprocation from a genuine responder, the
first mover's expected payoff from revealing is $qW - (1-q)L$ and from
withdrawing is $0$, where the signaling costs are sunk at this node. Revealing
is optimal iff $qW \ge (1-q)L$, that is $q(W+L) \ge L$, or
$q \ge L/(W+L) = \qstar$.

\emph{Monotonicity.} $\partial \qstar/\partial W = -L/(W+L)^2 < 0$ and
$\partial \qstar/\partial L = W/(W+L)^2 > 0$. Since
$W = (g-w)/(1-\delta_r) \to \infty$ as $\delta_r \to 1^-$ whenever $g > w$, we
have $\qstar \to 0$. \hfill$\square$

\subsection{Proof of Proposition~\ref{prop:platform}}

The initiator meets at least one responder with probability
$1-(1-\lambda)^n$, and conditional on meeting one obtains $qW-(1-q)L$ by
Proposition~\ref{prop:trust}. Posting is optimal iff
\[
  \left[1-(1-\lambda)^n\right]\left(qW-(1-q)L\right) \;\ge\; c_s + \mu\pi_s .
\]
Write $\Gamma = qW-(1-q)L$ and assume $\Gamma > c_s + \mu\pi_s$, which is
necessary for any $n$ to work. Rearranging,
$(1-\lambda)^n \le 1 - (c_s+\mu\pi_s)/\Gamma$, and taking logarithms with
$\ln(1-\lambda)<0$,
\[
  n \;\ge\; \frac{\ln\!\left(1-(c_s+\mu\pi_s)/\Gamma\right)}{\ln(1-\lambda)}
  \;=\; n^\star .
\]
Both $\mu$ and $\pi_s$ increase the numerator's argument toward $1$ from below,
hence increase $n^\star$. \hfill$\square$

\subsection{Proof of Proposition~\ref{prop:ringsize}}

With $m$ partners and per-review detection probability $d(m) = d_0 + \beta m$,
the per-round payoff is
\[
  \begin{aligned}
    \Pi(m) & = m\left[\rho b - \rho\left(\pi(d_0+\beta m) + c_e\right) - c_b\right] \\
           & = m\,\alpha - \rho\pi\beta m^2 ,
  \end{aligned}
\]
where $\alpha = \rho(b-c_e) - c_b - \rho\pi d_0$. This is strictly concave in
$m$ with $\Pi'(m) = \alpha - 2\rho\pi\beta m$, so the unconstrained maximizer is
$\mstar = \alpha/(2\rho\pi\beta)$, which is \eqref{eq:mstar}, and is positive
iff $\alpha > 0$. Then
$\partial \mstar/\partial c_e = -\rho/(2\rho\pi\beta) = -1/(2\pi\beta) < 0$,
$\partial \mstar/\partial \beta < 0$, and
$\partial \mstar/\partial b = 1/(2\pi\beta) > 0$. In practice $m$ is an integer
and bounded by the reviewer's load $L_i$, so the operative optimum is
$\min\{L_i, \lfloor \mstar \rceil\}$. \hfill$\square$

\subsection{Proof of Theorem~\ref{thm:dstar}}

By Theorem~\ref{thm:sustain}, cooperation is sustainable iff
$\delta_r \ge (d\pi + c_e)/b + c_b/(\rho b)$. Multiplying by $b$ and
rearranging,
\[
  \begin{aligned}
    d\pi        & \;\le\; \delta_r b - c_e - \frac{c_b}{\rho}                           \\
    \iff\quad d & \;\le\; \frac{1}{\pi}\left(\delta_r b - c_e - \frac{c_b}{\rho}\right)
    = \dstar(\delta_r).
  \end{aligned}
\]
Since $\dstar$ is increasing in $\delta_r$ and $\delta_r < 1$, the supremum over
admissible discount factors is $\dstar(1) = (b - c_e - c_b/\rho)/\pi$, above
which no discount factor sustains cooperation. \hfill$\square$

\section{Extended Simulation Details}
\label{app:sim}

Table~\ref{tab:params} lists the parameters that set the scale of the simulation
and the ranges swept, and the subsections below give the generative detail behind
them. Venue parameters are calibrated to published statistics, attack and game
parameters are swept over the stated ranges, and the five marked \emph{assumed}
are the ones no public source constrains. Appendix~\ref{app:sensitivity} varies
the assumed dispersions, which are the two that carry a conclusion.

\begin{table}[t]
  \centering
  \caption{Simulation parameters. Conference parameters are calibrated to
    published ICLR/NeurIPS statistics; game parameters are swept over the stated
    ranges except where a default is used.}
  \label{tab:params}
  \small
  \setlength{\tabcolsep}{4pt}
  \begin{tabular}{llp{2.4cm}}
    \toprule
    Parameter                             & Value          & Source                           \\
    \midrule
    \multicolumn{3}{l}{\emph{Venue, calibrated}}                                              \\
    Submissions $N$                       & 6000           & ICLR 2024 \cite{ref:openreview}  \\
    Reviewers $M$                         & 9000           & ICLR 2024 \cite{ref:openreview}  \\
    Reviews per paper $k$                 & 3              & venue policy                     \\
    Reviewer load $L_i$                   & 5              & venue policy                     \\
    Subject areas                         & 50             & venue call for papers            \\
    Score scale                           & 1--10          & venue policy                     \\
    Score noise $\sigma$                  & 1.2            & fitted \cite{ref:openreview}     \\
    Acceptance rate                       & 25\%           & ICLR 2024                        \\
    Bid weight (eager/willing)            & 1.0 / 0.5      & \cite{ref:tpms}                  \\
    Bid participation                     & 60\%           & calibration target               \\
    Eager bids per reviewer               & $\le 8$        & calibration target               \\
    \midrule
    \multicolumn{3}{l}{\emph{Attack, swept}}                                                  \\
    Colluder fraction                     & 1, 3, 5, 10\%  & swept                            \\
    Ring size $k$                         & 2, 4, 6, 8, 10 & swept                            \\
    Inflation $\Delta s$                  & $+\inflation$  & swept                            \\
    Partner affinity screen               & 97th pct.      & \S\ref{sec:e1}                   \\
    \midrule
    \multicolumn{3}{l}{\emph{Game, swept}}                                                    \\
    Detection prob.\ $d$                  & $[0, 0.30]$    & swept                            \\
    Effort cost $c_e$                     & $[0, 0.60]\,b$ & swept                            \\
    Trust $q$                             & $[0,1]$        & swept                            \\
    Discount $\delta_r$                   & 0.85           & 4--6 deadlines/yr                \\
    Sanction $\pi$                        & $5\,b$         & assumed                          \\
    Bid cost $c_b$                        & $0.02\,b$      & assumed                          \\
    Report loss $L$                       & $3\,b$         & assumed                          \\
    Spread of $b$, $\sigma_b$             & 0.35           & assumed, \S\ref{app:sensitivity} \\
    Spread of $\delta_r$, $\sigma_\delta$ & 0.10           & assumed, \S\ref{app:sensitivity} \\
    \bottomrule
  \end{tabular}
\end{table}

\subsection{Generating a conference}

Papers and reviewers occupy a shared simplex over $T = 50$ subject areas. Two
popularity vectors are drawn from a Dirichlet with concentration $1.2$, one for
submission volume and one for reviewer expertise, and the reviewer vector is
mixed with the paper vector at weight $0.7$. Reviewer supply therefore does not
track submission volume area by area, which produces over-bid and under-bid
regions of the conference. The under-bid region matters for the attack: it is
where a collusive bid faces little competition and wins the assignment on a
weaker margin than it would elsewhere.

Reviewer mixtures are Dirichlet draws around the reviewer popularity vector with
concentration $0.35$, low enough that reviewers are specialized. Authorship is
decided before paper content: a fraction $0.75$ of reviewers author one
submission each, and a paper's mixture is drawn with concentration $2.0$ around
a convex combination of its author's mixture, at weight $0.75$, and the global
popularity vector. Generating papers independently of their authors would place
an author's submission far from their own expertise, so a partner's paper would
be one the colluder could not plausibly review, and every strategy in
\S\ref{sec:taxonomy} would appear weaker than it is.

Affinity is the cosine similarity of the two mixtures plus Gaussian noise of
standard deviation $0.05$, clipped to $[0,1]$. It stands in for the
text-matching score a venue computes from publication records. Latent quality
$\theta_p$ is standard normal and independent of topics, which is what allows
\S\ref{sec:impact} to measure displacement: no observational dataset contains
ground-truth quality. The only conflict of interest the venue detects
automatically is authorship of the paper under review. Dyad partners are
strangers by Definition~\ref{def:dyad}, so no declared conflict exists between
them, which is what the attack relies on.

\subsection{Honest bidding}

Each reviewer inspects the $60$ submissions with the highest affinity to their
own profile. Interest in an inspected paper is its affinity plus Gaussian noise
of standard deviation $0.12$, so two reviewers with the same profile do not bid
on the same papers. A reviewer bids \emph{eager} on an inspected paper with
probability $\sigma(12(\text{interest} - 0.55))$ where $\sigma$ is the logistic
function, up to a cap of $8$ eager bids, and the next papers by interest become
\emph{willing} bids. A fraction $0.40$ of reviewers place no bids at all.

The participation rate and the cap are calibration targets rather than
modeling conveniences. If every high-affinity reviewer bid on every paper in
their area, a bid would confer no advantage in the assignment and there would be
nothing for a colluder to manipulate. The targets are a mean near six eager bids
per paper with roughly a quarter of papers receiving two or fewer, which
reproduces the concentration reported for real venues.

\subsection{Injecting collusion}

Only reviewers who author a submission are eligible, since a colluder needs a
paper to trade. Rings are formed among the eligible pool by a screen on mutual
affinity: writing $A_{ij}$ for reviewer $i$'s affinity to $j$'s submission, the
pair $(i,j)$ is admissible when the symmetrized affinity
$\tfrac{1}{2}(A_{ij} + A_{ji})$ exceeds the $97$th percentile of each side's own
affinity distribution. Rings are then grown greedily from a random order,
sampling each next member among the eight most similar admissible candidates
rather than taking the most similar, so that rings are not deterministic in the
seed.

The screen is a substantive constraint rather than a detail. A recruitment post
names a subject area and the people who answer one work in it, so partners are
topically adjacent while remaining strangers in the sense of
Definition~\ref{def:dyad}. Pairing uniformly at random would understate every
strategy, because a bid on a topically distant paper loses the assignment on
affinity alone and is conspicuous besides. The screen is also why not every
willing reviewer finds a usable partner, and why the number of completed rings
falls as the ring grows.

Which members a colluder bids on follows the strategy. Under the naive,
camouflaged and affinity-aware strategies every member bids on every other
member's submission; under the distributed and combined strategies the ring is a
directed cycle in which each member bids on exactly one successor, so the bid
graph contains no two-cycle.

Camouflaging strategies add $r$ cover bids per collusive bid, $r = 4$ by default.
Cover papers are ranked by $A_{i\cdot} + 1.5\,s_{\cdot}$, where $s$ is the
maximum topical similarity to any targeted paper. Both terms are necessary.
Ranking on similarity to the target alone drags the whole bid set away from the
reviewer's area and makes the profile more anomalous, not less; ranking on the
reviewer's own affinity alone leaves the target as isolated as it was. The
bridging term is what stops the collusive bid from being an outlier within its
own bid set. Appendix~\ref{app:detection} reports the effect of $r$.

Affinity-aware strategies add $0.15$ to the affinity of a targeted pair, clipped
at one, and the assignment is solved against the manipulated matrix. This models
profile curation and matched background text as an additive boost rather than as
an actual editing procedure, which is a limitation: it fixes the gain a colluder
achieves instead of deriving it from what the text permits.

\subsection{Assignment}

Problem~\eqref{eq:assignment} is solved as a linear program over a sparsified
candidate set: the $30$ highest-scoring reviewers per paper under
affinity-plus-bid-weight, together with every reviewer-paper pair carrying a bid,
less all authorship pairs. Retaining every bid pair is required for validity.
Dropping a low-affinity collusive bid during sparsification would defend against
the attack being measured, and would do so silently.

Coverage is an equality constraint at $k = 3$ reviews per paper, load an
inequality at $L_i = 5$, and edge variables are bounded by a per-pair cap. The
constraint matrix of a bipartite $b$-matching is totally unimodular, so the
relaxation has an integral optimum; HiGHS returns one and the assignment is read
off at $x > 0.5$.

Randomized assignment keeps the fractional solution with the per-pair cap set to
$0.35$ and samples an integral assignment from it, processing papers in random
order and drawing reviewers without replacement from the fractional marginals
while skipping reviewers already at their load. The cap is therefore enforced
approximately rather than exactly, since the repair needed for feasibility can
perturb a marginal. Cycle-free assignment is a post-hoc repair that finds pairs
of reviewers assigned to each other's submissions and drops the lower-scoring of
the two assignments. It is a heuristic at $k=2$, not an exact integer program;
the exact formulation is NP-hard, and \S\ref{sec:e4} reports it as implemented.

\subsection{Reviewing and decisions}

An assigned reviewer reports $5.5 + 1.5\,\theta_p + 1.2\,\varepsilon$ clipped to
the $1$ to $10$ scale. The noise term $\varepsilon$ is not drawn sequentially. It
is a SplitMix64 hash of the triple (seed, reviewer, paper) mapped to a standard
normal through the Box-Muller transform, so an honest reviewer's score on an
honest paper is identical across runs. Sequential draws from a shared generator
would not have this property: injecting collusion consumes a different number of
draws, the honest scores would diverge between a collusive run and its baseline,
and every displacement measurement in \S\ref{sec:impact} would be contaminated
by that divergence rather than by the collusion.

A collusive review occurs wherever a colluder is assigned to a paper they placed
a collusive bid on. Uniform maximum sets the score to the top of the scale,
calibrated inflation adds $\Delta s = \inflation$, and constructive positive adds
$0.8\,\Delta s$. Author-drafted reviewing is not a fourth scoring arm: it differs
from constructive positive in effort cost and in resistance to content-based
scrutiny rather than in the score it produces, so it enters the analysis through
$c_e$ in \S\ref{sec:equilibrium} and not through the score distribution.
Decisions accept the top $25$\% of submissions by mean score, with unreviewed
papers ranked last.

\subsection{Detectors and operating points}

\emph{Bid frequency} scores a reviewer by the isolation of their most isolated
eager bids. Isolation of a bid is the negative of its maximum topical similarity
to any other paper in the same bid set, plus the shortfall of its affinity below
the reviewer's own $90$th percentile, and the reviewer's score is the sum of the
two largest. This is a stronger test than counting low-affinity bids, and it is
the reason a camouflaging colluder must draw cover bids from around the target
rather than from their own usual area.

\emph{Dense subgraph} builds the reviewer-author graph, joining $r$ and $a$ when
$r$ bids on a paper $a$ authored and $a$ also reviews, and scores a reviewer by
the edge density of the subgraph induced by their neighborhood, weighted by
$1 + $ the number of neighbors who bid back. We score locally rather than
running Charikar peeling globally because peeling returns one densest subgraph,
and the object of interest here is a population of many small disjoint rings.

\emph{Cycle detection} orients the same graph and counts short directed cycles
through each reviewer. The default bound of $2$ is the constraint that deployed
cycle-free assignment enforces; Appendix~\ref{app:detection} reports bounds $3$
and $4$.

\emph{Text anomaly} scores an assignment by how far its realized affinity falls
below what the reviewer normally receives, taking the reviewer's own top-$L_i$
affinities as the reference distribution and summing the positive $z$-scores over
their assigned papers. The reference has to be the top of the distribution rather
than the reviewer's mean, since almost any assignment beats a reviewer's average
paper and a mean-based test flags nothing.

Every family is thresholded at the $1 - \text{FPR}$ quantile of its scores over
honest reviewers, which makes the families comparable at a fixed investigation
budget. That is the right common currency because a program chair's binding
constraint is how many reviewers they can afford to examine, not any intrinsic
property of a score. Ties at zero are handled explicitly, by flagging only scores
strictly above the threshold, so a family that assigns the same score to most of
the pool receives no credit for flagging all of them. The implementations are
deliberately reasonable rather than deliberately weak: where a family admits an
obvious strengthening we take it, because the claim being made is that these
signals fail against specific countermeasures, and that claim is only worth
making against detectors that work on the naive adversary.

\subsection{Determinism and cost}

Conference generation, honest bidding, collusion injection and assignment draw
from independent substreams of one seed sequence per run, and review noise is
hashed as described above. Honest bids are therefore identical whether or not
collusion is injected. A collusive run and its baseline share the generated
conference, so the two arms differ in the injected bids and in nothing else. One
conference at the reported size takes about $2.5$ seconds and peaks near
$1.8$ GB; the full suite of experiments reported here runs in about $20$ minutes
on one core, and re-running it reproduces every macro in the paper exactly.

\section{Sensitivity Analyses}
\label{app:sensitivity}

Four inputs to \S\ref{sec:experiments} are set by assumption rather than by
evidence: the dispersion of $b$ and $\delta_r$ across the population, the false
positive rate at which the detector families are compared, the acceptance rate,
and the per-reader response rate of Proposition~\ref{prop:platform}. This
appendix varies each in turn and reports which conclusions move.

\subsection{Heterogeneity in $b$ and $\delta_r$}

Table~\ref{tab:sens-het} varies the two spreads. Setting both to zero is the
informative case, because it isolates the explanation \S\ref{sec:e1} gives for
the distance between the measured transitions and the closed forms. That
explanation is congestion, and if it is right the distance should survive an
identical population. It does, and it widens: the $d$ transition falls to
$\sensDTransHomog$ against $\dstar = \dStar$. The mechanism is visible in the
same row. With one agent type every agent makes the same choice, so wherever
collusion is sustainable at all the fixed point is $\sensFracHomog$\% of the
population, and the congestion multiplier is evaluated at its maximum.
Dispersion spreads agents across their individual thresholds, holds the fixed
point below one, and weakens the congestion term.

\begin{table}[t]
  \centering
  \caption{Measured transitions and effort-regime collusion levels against the
    spread of $b$ and of $\delta_r$, over $\nAgents$ agents. The three rightmost
    columns are the quantities \S\ref{sec:e1} quotes. The closed forms are evaluated
    at the population mean and so do not move with either spread.}
  \label{tab:sens-het}
  \input{generated/sens_het}
\end{table}

Across every setting swept, the $d$ transition stays within
$[\sensDTransSpreadLo, \sensDTransSpreadHi]$ and the $q$ transition within
$[\sensQTransSpreadLo, \sensQTransSpreadHi]$, so neither panel of
Fig.~\ref{fig:equilibrium} depends on the assumed dispersion. The $c_e$
transition does, ranging over
$[\sensCeTransSpreadLo, \sensCeTransSpreadHi]$, and almost all of that movement
comes from the spread of $b$ rather than of $\delta_r$. The reason is in
\eqref{eq:deltastar}: effort enters as $c_e/b$, so dispersion in $b$ rescales
the effective effort cost across the population.

\S\ref{sec:e1} quotes no $c_e$ transition, but it does quote the collusion level
at three effort regimes, and those move: hand-written effort sustains
$\sensRegimeHandHi$\% of the population at $\sigma_b = 0$ and
$\sensRegimeHandLo$\% at $\sigma_b = \sensBSpreadHi$. The ordering that
Corollary~\ref{cor:comparative} asserts survives every configuration, including
the homogeneous one, where all three regimes saturate and the comparative static
is invisible at the population level while remaining exactly as
Corollary~\ref{cor:comparative} states it pointwise. What the spread of $b$ sets
is the magnitude of the effort effect, not its direction, and the percentages in
\S\ref{sec:e1} should be read as conditional on $\sigma_b = 0.35$.

\subsection{Detector operating points}

The main text compares the four families at a single $\evasionFPR$\% false
positive rate. Table~\ref{tab:sens-fpr} gives the best $F_1$ any of the four
attains, at budgets from $\gridFPRLo$\% to $\gridFPRHi$\% of the honest pool.

\begin{table}[t]
  \centering
  \caption{Best $F_1$ over the four detector families, by false positive rate,
    mean over $\gridSeeds$ instances. Rows are bidding strategies. The operating
    point changes only the threshold, so detector scores are computed once and
    re-thresholded.}
  \label{tab:sens-fpr}
  \input{generated/grid_fpr}
\end{table}

The bound the paper claims is not an artifact of the operating point. Against
the combined strategy the best $F_1$ over all four families and all five budgets
is $\gridBestCombAnywhere$, below the $\fOneBestOverall$ of
Table~\ref{tab:detection}, and the distributed strategy stays lower still. A
ring longer than the cycle bound produces no cycle to count at any threshold,
and camouflage leaves nothing for the frequency test to find at any threshold, so
loosening the budget does not recover either.

The detectability of the \emph{unsophisticated} attacker does depend on the
budget, sharply. Naive bidding moves from $F_1 = \fOneNaiveCycle$ at
$\evasionFPR$\% to $\gridBestAnywhere$ at $10$\%. The jump is a threshold
effect with a specific cause: only $\cycleHonestShallowPct$\% of honest reviewers
carry any two-cycle at all (Table~\ref{tab:cycle}), so a budget below that
fraction places the threshold inside the honest nonzero mass, while a budget
above it places the threshold at zero and flags every reviewer holding a single
cycle. The $\evasionFPR$\% convention therefore sits just below a discontinuity
in what cycle detection can do. A venue willing to examine one reviewer in ten
identifies naive and camouflaged dyads at more than triple the $F_1$ of
Table~\ref{tab:detection}, which is worth knowing, and gains nothing against the
strategies the recruitment threads actually describe.

\subsection{Acceptance rate}

The accept/reject rule is a top-$k$ on realized paper means, so the rate can be
varied after the scores are fixed. Table~\ref{tab:sens-accept} does that from
one collusive run and one baseline run per instance, which means the arms differ
in the threshold and in nothing else; the $25$\% row reproduces
\S\ref{sec:impact} exactly.

\begin{table}[t]
  \centering
  \caption{Impact against the acceptance rate, camouflaged dyads at
    $\colluderFrac$\% of the pool, mean over $\impactSeeds$ instances. Boost is in
    percentage points over the same colluders' baseline rate.}
  \label{tab:sens-accept}
  \input{generated/sens_acceptance}
\end{table}

No rate in the range removes the effect. The boost grows with the acceptance
rate, from $\sensBoostRateLo$ percentage points at $\sensRateLo$\% to
$\sensBoostRateHi$ at $\sensRateHi$\%, and the count of displaced honest papers
grows with it as well, from $\sensDispRateLo$ to $\sensDispRateHi$, since a
larger accepted set offers more decisions to disturb. Selectivity therefore
reduces the absolute harm a fixed population of colluders can do while leaving
the colluders' advantage intact, and a venue at $\sensRateLo$\% still admits
$\sensAcceptRateLo$\% of colluding submissions against
$\sensAcceptRateHi$\% at $\sensRateHi$\%.

\subsection{Audience size}

Proposition~\ref{prop:platform} leaves the per-reader response rate $\lambda$
free, and \S\ref{sec:equilibrium} asserts that a public group clears
$n^{\star}$ at any plausible value. Table~\ref{tab:sens-platform} quantifies
that. Even at $\sensLambdaLo$\%, one reply per hundred readers,
$n^{\star} = \nStarLambdaLo$ readers, and at $\sensLambdaHi$\% it is
$\nStarLambdaHi$. A channel with tens of members clears the threshold by orders
of magnitude, so audience size is not what restrains posting. The binding
condition is $\Gamma > c_s + \mu\pi_s$, that a partnership be worth more than
the cost of soliciting one, which is where a venue's monitoring intensity $\mu$
enters and why it is the only lever in the model that acts before any bid is
placed.

\begin{table}[t]
  \centering
  \caption{Minimum audience $n^{\star}$ of Proposition~\ref{prop:platform} against
    the per-reader response rate $\lambda$, at the default payoffs. Values below one
    mean a single expected reply justifies the post.}
  \label{tab:sens-platform}
  \input{generated/sens_platform}
\end{table}

\subsection{Tolerance frontier}

Fig.~\ref{fig:frontier} plots the feasibility frontier of
Proposition~\ref{prop:tolerance} over $(\rho, \delta_r)$. The region in which no
$k$-strike rule both deters free-riding and survives unlucky runs of
non-assignment is exactly the region a venue is trying to reach by lowering
$\rho$, and it widens as $\delta_r$ falls.

\begin{figure}[t]
  \centering
  \includegraphics[width=\columnwidth]{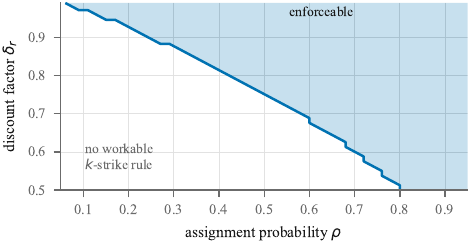}
  \caption{Feasibility frontier for the $k$-strike rule of
    Proposition~\ref{prop:tolerance}. Inside the shaded region
    $k_{\min} \le k_{\max}$ and the partnership can police itself; outside it, no
    tolerance setting deters a partner who quietly withholds their bid without
    also destroying honestly cooperating partnerships.}
  \label{fig:frontier}
\end{figure}

\section{Additional Detection Results}
\label{app:detection}

Table~\ref{tab:evasion} states the outcome the model predicts for each pairing
of strategy and detector family. Table~\ref{tab:detection} in
\S\ref{sec:e2} reports the measured counterpart.

\begin{table}[t]
  \centering
  \caption{Strategy against detector family. \checkmark~denotes reliable evasion,
    $\sim$~partial evasion, \ding{55}~detection. Entries are predictions from the
    model; measured $F_1$ appears in Table~\ref{tab:detection}.}
  \label{tab:evasion}
  \small
  \setlength{\tabcolsep}{4pt}
  \begin{tabular}{lcccc}
    \toprule
                               & \rotatebox{55}{Bid freq.}  & \rotatebox{55}{Dense subgr.}
                               & \rotatebox{55}{Cycle det.} & \rotatebox{55}{Text anom.}                                             \\
    \midrule
    Naive                      & \ding{55} & \ding{55}   & \ding{55} & $\sim$     \\
    Camouflaged                & \checkmark                 & $\sim$                       & \ding{55} & $\sim$     \\
    Distributed                & \checkmark                 & $\sim$                       & \checkmark                 & $\sim$     \\
    Affinity-aware             & $\sim$                     & \ding{55}   & \ding{55} & \checkmark \\
    Camo.\ + dist.\ + affinity & \checkmark                 & \checkmark                   & \checkmark                 & \checkmark \\
    \bottomrule
  \end{tabular}
\end{table}

\subsection{Precision and recall}

Table~\ref{tab:grid} separates the two errors that Table~\ref{tab:detection}
compresses into $F_1$. The distinction is not cosmetic, because two entries in
the table would be read differently by a program chair than their $F_1$ suggests.
The bid-frequency family reaches precision $\gridFreqDistPrec$ against both the
distributed and the combined strategy, at recall $\gridFreqDistRec$ and
$\gridFreqCombRec$. When it fires it is never wrong, and it almost never fires.
Its $F_1$ of $\gridFreqDistFOne$ and $\gridFreqCombFOne$ invites discarding a
signal that at this operating point produces no false accusations at all, which
is the error a venue is least able to afford. Cycle detection against camouflage
is the opposite shape, precision $\gridCycleCamoPrec$ at recall
$\gridCycleCamoRec$, and it holds the highest $F_1$ in the table.

\begin{table*}[t]
  \centering
  \caption{Precision, recall and $F_1$ for every strategy against every detector
    family, at a common $\evasionFPR$\% false positive rate, mean over $\gridSeeds$
    instances. Table~\ref{tab:detection} reports the $F_1$ columns.}
  \label{tab:grid}
  \input{generated/grid_full}
\end{table*}

\subsection{How much camouflage}

\S\ref{sec:e2} notes that camouflage raises the cycle-detection score.
Table~\ref{tab:camo} gives the dose-response and shows the effect is monotone
over the whole range. As cover bids per collusive bid rise from $1$ to
$\camoRatioHi$, the collusive share of the bid vector falls from $\camoShareHi$
to $\camoShareLo$, which is what camouflage is for, while $F_1$ rises for both
graph-based families, from $\camoDenseLo$ to $\camoDenseHi$ for dense subgraph
and from $\camoCycleLo$ to $\camoCycleHi$ for cycle detection. The frequency test
stays at $\camoFreqHi$ throughout and the text test is unmoved.

\begin{table}[t]
  \centering
  \caption{Camouflage dose-response: cover bids per collusive bid against
    detection $F_1$ at $\evasionFPR$\% false positive rate and assignment success,
    dyads at $\colluderFrac$\% of the pool, mean over $\gridSeeds$ instances. Honest
    reviewers place $\bidsHonest$ eager bids on average. Ratio $0$ is the naive
    strategy.}
  \label{tab:camo}
  \input{generated/grid_camouflage}
\end{table}

Two mechanisms drive the graph columns. Each cover bid adds an edge to the
reviewer-author graph, so a heavily camouflaged colluder has more chances to
close an incidental cycle with an honest author; and the bid vector grows to
$\camoBidsHi$ entries against an honest mean of $\bidsHonest$, so the
neighborhood whose density the dense-subgraph test measures is larger than any
honest reviewer's. Camouflage also costs assignment success, which falls from
$\camoSuccHi$ to $\camoSuccLo$, because cover bids are eager bids competing for
the same reviewer load.

Camouflage therefore has an interior optimum and it lies at the low end. The
attacker's best choice is the smallest ratio that zeroes the frequency test,
which is the smallest ratio we test, and heavier camouflage is worse on every
axis measured: more visible to two families, no less visible to the other two,
and less likely to win the assignment. The description of camouflage as cheap in
\S\ref{sec:taxonomy} holds at the low ratio; the assumption that more of it is
better does not. We also note a signal none of the four families uses. Bid volume
is what heavy camouflage inflates, and a test on volume alone would separate the
$\camoRatioHi\times$ row from the honest population immediately. That is a
concrete instance of the multi-signal argument of \S\ref{sec:discussion}.

\subsection{Cycle detection beyond two}

A distributed ring evades a two-cycle constraint by construction, and the obvious
reply is to raise the bound. Table~\ref{tab:cycle} sweeps bounds
$2$ through $\cycleBoundHi$ against rings of $2$ through $\cycleRingHi$ members,
reporting $F_1$ alongside the fraction of honest reviewers who carry at least one
cycle at that bound.

\begin{table}[t]
  \centering
  \caption{Cycle detection at bound $L$ against ring size, distributed bidding at
    $\colluderFrac$\% of the pool, $\evasionFPR$\% false positive rate, mean over
    $\gridSeeds$ instances. \emph{honest} is the fraction of honest reviewers holding
    at least one cycle at that bound, which is what limits the threshold.}
  \label{tab:cycle}
  \input{generated/grid_cycle}
\end{table}

Raising the bound works when it matches the ring. A three-member ring goes from
$\cycleRingThreeShallow$ at bound $2$, which is evasion, to
$\cycleRingThreeMatched$ at bound $3$. It does not keep working. The same ring
falls to $\cycleRingThreeDeep$ at bound $\cycleBoundHi$, and a four-member ring
reaches only $\cycleRingFourMatched$ even though the bound matches it. The reason
is in the honest column: the share of honest reviewers carrying a cycle rises
from $\cycleHonestShallow$ at bound $2$ to $\cycleHonestDeep$ at bound
$\cycleBoundHi$. Longer cycles are ordinary in honest bidding, so a longer bound
admits far more honest mass than collusive signal, and a fixed false positive
budget then buys a threshold that separates less than it did.

The exchange favors the attacker on both sides. Lengthening the ring costs one
additional partner from the same recruitment thread, which
Proposition~\ref{prop:ringsize} prices, and it is available to anyone who can
read the thread. Matching it costs the venue an enumeration whose time per
instance grows from $\cycleSecShallow$ to $\cycleSecDeep$ seconds here and grows
faster than that with the bound, together with a false positive base that more
than doubles. At bound $\cycleBoundHi$ a ring of $\cycleRingHi$ still sits at
$\cycleFOneRingSixDeep$. Cycle detection is a bound-guessing game, and raising
the bound alone does not win it.

\subsection{On temporal dispersion}

A strategy that spreads collusive bids across the bidding period is not reported
here, and the reason is a property of the detector families rather than a gap in
the sweep. None of the four uses bid timing: each reads the bid set, the bid
graph, or the affinity of the resulting assignment, and all three are invariant
to the order in which bids arrive. A dispersed variant of any strategy in
Table~\ref{tab:grid} would reproduce its row exactly. Measuring dispersion needs
a fifth family built on timing and a calibrated arrival process for honest bids,
and no venue publishes bid timestamps. We record it in
Appendix~\ref{app:limitations} as future work, and note that arrival time is one
of the few coordination problems a pair cannot solve without communicating
further, so a timing-based family would attack the part of the arrangement that
\S\ref{sec:model} identifies as most exposed.

\section{Ethics Statement}
\label{app:ethics}

This appendix expands \S\ref{sec:discussion}.

\emph{Data collection.} The recruitment threads in Fig.~\ref{fig:screenshots}
are public. We collected no private messages, joined no groups under false
pretenses, and conducted no interaction with the people posting. The replies
shown are a selection: each thread carries more than we reproduce, and we kept
those that bear on the model.

\emph{Reproduction.} We typeset the threads rather than reproduce screenshots,
because the platform is identifiable from its interface chrome alone. The text
is translated and lightly paraphrased so that no thread can be recovered by
searching a quoted string, which a verbatim transcription would still permit
after the images were removed. Handles, avatars, timestamps, group identifiers,
and venue names are redacted, we name no individuals or groups, and we report
neither the platform nor the search terms that locate the threads. The threads
serve as evidence that the phenomenon exists and as motivation for the model; no
part of the analysis depends on their content.

\emph{Disclosure.} The attack taxonomy of \S\ref{sec:taxonomy} describes
strategies more effective than those the published detector evaluations assume.
We include it because a defense cannot be evaluated against an unspecified
adversary, and because the strategies are not novel to us: camouflage and ring
topologies are discussed openly in the recruitment channels themselves, and
affinity manipulation is published~\cite{ref:hsieh2025}. The disclosure balance
is the standard one in security research, and we judge it to favor publication.

\emph{Sting operations.} Lowering $q$, the probability that a responder is
genuine, appears in the model as a lever and amounts to a venue running sting
operations against its own reviewers. We note it for completeness in
\eqref{eq:qstar} and do not recommend it. It requires a venue to deceive its
volunteers, it would poison the recruitment channels in ways that are hard to
reason about, and Proposition~\ref{prop:trust} shows its effect is bounded by
$\qstar$ already being low for durable partnerships.

\section{Extended Limitations and Future Work}
\label{app:limitations}

The repeated game is between a fixed pair across rounds. Real partnerships form,
dissolve, and re-form through the platform, and a network formation model would
capture the reputation mechanism we describe informally in
\S\ref{sec:equilibrium} rather than folding it into $\delta_r$. We also model a
single venue: cross-venue collusion, in which a ring operates simultaneously at
several conferences that share no data, both raises the effective number of
rounds and lowers the effective detection probability, and we discuss it without
modeling it. Finally, we treat detection probability $d$ as exogenous, whereas a
venue that deploys a detector faces an adversary who adapts, and the fixed point
of that interaction is not the fixed point we compute.

Several parts of this paper are deliberately shallow and are being taken up
separately: a multi-signal detection framework built on the observation that the
countermeasures in \S\ref{sec:taxonomy} are
mutually constraining; an adaptive adversary analysis in which detector and
strategy co-evolve; the expertise-familiarity trade-off, which governs how much
room a venue has to exclude suspicious matches before review quality suffers;
cross-venue collusion, which our model says should be strictly easier to sustain
than single-venue collusion; and the empirical question of whether
author-drafted reviews leave a stylometric signature.

Two specific detector questions are left open by
Appendix~\ref{app:detection}. Bid timing is a signal none of the four families
uses, and it is the one dimension along which a pair cannot coordinate without
communicating further, so a timing-based family would bear on the stage of the
arrangement the model identifies as most exposed. Evaluating one requires a
calibrated arrival process for honest bids, which no venue publishes. Bid volume
is a signal our frequency family also does not use, and it is what heavy
camouflage inflates; the dose-response in Table~\ref{tab:camo} suggests that
volume and topical isolation together constrain the attacker more tightly than
either does alone, which is the multi-signal argument stated for a case where we
can measure both sides of it.

%% file: generated/sens_het.tex
% Generated by sim/plots/make_numbers.py. Do not edit by hand.
\small
\setlength{\tabcolsep}{4pt}
\begin{tabular}{lccccc}
\toprule
& $d$ & $q$ & \multicolumn{3}{c}{Collusion at $c_e$ (\%)} \\
\cmidrule(lr){4-6}
Spread & trans. & trans. & hand & machine & author \\
\midrule
\multicolumn{6}{@{}l}{\emph{no heterogeneity in either}} \\
$0$ & 0.060 & 0.702 & 100.0 & 100.0 & 100.0 \\
\addlinespace
\multicolumn{6}{@{}l}{\emph{spread of $b$, holding $\sigma_{\delta} = 0.10$}} \\
$0.000$ & 0.081 & 0.641 & 87.4 & 99.8 & 100.0 \\
$0.175$ & 0.080 & 0.640 & 79.5 & 98.1 & 99.3 \\
$0.350$ & 0.079 & 0.627 & 68.2 & 91.9 & 95.2 \\
$0.700$ & 0.076 & 0.573 & 51.9 & 75.1 & 80.8 \\
\addlinespace
\multicolumn{6}{@{}l}{\emph{spread of $\delta_r$, holding $\sigma_b = 0.35$}} \\
$0.000$ & 0.076 & 0.650 & 69.0 & 95.1 & 98.1 \\
$0.050$ & 0.077 & 0.642 & 68.9 & 94.1 & 97.3 \\
$0.100$ & 0.079 & 0.627 & 68.2 & 91.9 & 95.2 \\
$0.200$ & 0.081 & 0.585 & 65.4 & 86.5 & 89.9 \\
\addlinespace
\multicolumn{6}{@{}l}{\emph{closed form, at the population mean}} \\
any & 0.113 & 0.571 & --- & --- & --- \\
\bottomrule
\end{tabular}

%% file: generated/grid_fpr.tex
% Generated by sim/plots/make_numbers.py. Do not edit by hand.
\small
\setlength{\tabcolsep}{5pt}
\begin{tabular}{lccccc}
\toprule
Strategy & $1$\% & $2$\% & $5$\% & $10$\% & $20$\% \\
\midrule
Naive & 0.139 & 0.156 & 0.156 & 0.567 & 0.567 \\
Camouflaged & 0.261 & 0.322 & 0.322 & 0.551 & 0.551 \\
Distributed & 0.079 & 0.079 & 0.110 & 0.145 & 0.158 \\
Affinity-aware & 0.139 & 0.156 & 0.156 & 0.567 & 0.567 \\
Camo.+dist.+affinity & 0.127 & 0.158 & 0.216 & 0.228 & 0.221 \\
\bottomrule
\end{tabular}

%% file: generated/sens_acceptance.tex
% Generated by sim/plots/make_numbers.py. Do not edit by hand.
\small
\setlength{\tabcolsep}{5pt}
\begin{tabular}{lcccc}
\toprule
Acceptance & Colluder & Boost & Displaced & Undeserved \\
rate & rate (\%) & (pp) & papers & papers \\
\midrule
$15$\% & 17.4 & +2.4 & 39 & 43 \\
$20$\% & 23.0 & +3.1 & 46 & 50 \\
$25$\% & 28.2 & +3.2 & 51 & 55 \\
$30$\% & 33.0 & +3.5 & 59 & 64 \\
$40$\% & 43.4 & +4.7 & 68 & 72 \\
\bottomrule
\end{tabular}

%% file: generated/sens_platform.tex
% Generated by sim/plots/make_numbers.py. Do not edit by hand.
\small
\setlength{\tabcolsep}{4pt}
\begin{tabular}{lc}
\toprule
Response rate $\lambda$ & $n^{\star}$ (readers) \\
\midrule
$1$\% & 4.88 \\
$2$\% & 2.43 \\
$5$\% & 0.96 \\
$10$\% & 0.47 \\
$20$\% & 0.22 \\
\bottomrule
\end{tabular}

%% file: generated/grid_full.tex
% Generated by sim/plots/make_numbers.py. Do not edit by hand.
\small
\setlength{\tabcolsep}{5pt}
\begin{tabular}{lcccccccccccc}
\toprule
& \multicolumn{3}{c}{Bid frequency} & \multicolumn{3}{c}{Dense subgraph} & \multicolumn{3}{c}{Cycle detection} & \multicolumn{3}{c}{Text anomaly} \\
\cmidrule(lr){2-4} \cmidrule(lr){5-7} \cmidrule(lr){8-10} \cmidrule(lr){11-13}
Strategy & P & R & $F_1$ & P & R & $F_1$ & P & R & $F_1$ & P & R & $F_1$ \\
\midrule
Naive & 0.000 & 0.000 & 0.000 & 0.063 & 0.075 & 0.069 & 0.362 & 0.100 & 0.156 & 0.043 & 0.052 & 0.047 \\
Camouflaged & 0.000 & 0.000 & 0.000 & 0.148 & 0.185 & 0.164 & 0.530 & 0.232 & 0.322 & 0.028 & 0.033 & 0.030 \\
Distributed & 1.000 & 0.041 & 0.079 & 0.098 & 0.127 & 0.110 & 0.092 & 0.019 & 0.032 & 0.044 & 0.056 & 0.049 \\
Affinity-aware & 0.000 & 0.000 & 0.000 & 0.063 & 0.075 & 0.069 & 0.362 & 0.100 & 0.156 & 0.077 & 0.094 & 0.085 \\
Camo.+dist.+affinity & 1.000 & 0.068 & 0.127 & 0.186 & 0.261 & 0.216 & 0.254 & 0.075 & 0.115 & 0.056 & 0.071 & 0.063 \\
\bottomrule
\end{tabular}

%% file: generated/grid_camouflage.tex
% Generated by sim/plots/make_numbers.py. Do not edit by hand.
\small
\setlength{\tabcolsep}{3pt}
\begin{tabular}{lccccccc}
\toprule
Cover & Collusive & Eager & \multicolumn{4}{c}{$F_1$ by family} & Assign. \\
\cmidrule(lr){4-7}
bids & share & bids & freq. & dense & cycle & text & success \\
\midrule
$1\times$ & 0.258 & 6.7 & 0.000 & 0.096 & 0.205 & 0.038 & 0.296 \\
$2\times$ & 0.186 & 7.6 & 0.000 & 0.117 & 0.251 & 0.032 & 0.285 \\
$4\times$ & 0.124 & 9.5 & 0.000 & 0.164 & 0.322 & 0.030 & 0.271 \\
$8\times$ & 0.078 & 13.4 & 0.000 & 0.245 & 0.416 & 0.034 & 0.232 \\
\bottomrule
\end{tabular}

%% file: generated/grid_cycle.tex
% Generated by sim/plots/make_numbers.py. Do not edit by hand.
\small
\setlength{\tabcolsep}{4pt}
\begin{tabular}{lcccccc}
\toprule
& \multicolumn{2}{c}{$L=2$} & \multicolumn{2}{c}{$L=3$} & \multicolumn{2}{c}{$L=4$} \\
\cmidrule(lr){2-3} \cmidrule(lr){4-5} \cmidrule(lr){6-7}
Ring size & $F_1$ & honest & $F_1$ & honest & $F_1$ & honest \\
\midrule
$2$ & 0.156 & 0.068 & 0.228 & 0.119 & 0.138 & 0.180 \\
$3$ & 0.032 & 0.067 & 0.248 & 0.118 & 0.144 & 0.180 \\
$4$ & 0.041 & 0.068 & 0.106 & 0.118 & 0.141 & 0.180 \\
$6$ & 0.031 & 0.068 & 0.098 & 0.119 & 0.095 & 0.180 \\
\bottomrule
\end{tabular}